\pdfoutput=1
\documentclass[pra,aps,superscriptaddress,nofootinbib,twocolumn,longbibliography]{revtex4-2}
\usepackage{latexsym,amsmath,amssymb,amsfonts,graphicx,color,amsthm}
\usepackage{enumerate}
\usepackage{enumitem}
\usepackage{braket}
\usepackage{adjustbox}
\usepackage{footnote}
\usepackage[dvipsnames]{xcolor}
\usepackage{tipa}
\usepackage{textcomp}
\usepackage{fontenc}
\usepackage{mathtools}
\usepackage{mathrsfs}
\usepackage{color}
\usepackage{colortbl}
\usepackage{graphics,graphicx}
\usepackage{txfonts}
\usepackage{lipsum, babel}
\usepackage{bbm}
\usepackage{amsmath}
\usepackage{amssymb}
\usepackage[unicode=true,pdfusetitle, bookmarks=true,bookmarksnumbered=false,bookmarksopen=false, breaklinks=false,pdfborder={0 0 0},backref=false,colorlinks=false] {hyperref}
\hypersetup{ colorlinks,linkcolor=myurlcolor,citecolor=myurlcolor,urlcolor=myurlcolor}

\definecolor{myurlcolor}{rgb}{0,0,0.7}

\usepackage{float}
\usepackage{hyperref}
\newcommand{\be}{\begin{equation}}
\newcommand{\ee}{\end{equation}}
\newcommand{\bea}{\begin{eqnarray}}
\newcommand{\eea}{\end{eqnarray}}

\providecommand{\conv}{\operatorname{conv}}
\providecommand{\Cset}{\mathcal{C}}
\providecommand{\QCset}{\mathcal{Q}^{C}}

\newcommand{\cC}{\mathcal{C}}

\newcommand{\cH}{\mathcal{H}}
\newcommand{\cI}{\mathcal{I}}
\newcommand{\cJ}{\mathcal{J}}
\newcommand{\cK}{\mathcal{K}}
\newcommand{\cL}{\mathcal{L}}

\newcommand{\cQ}{\mathcal{Q}}

\newcommand{\cS}{\mathcal{S}}

\newcommand{\Id}{\mathbbm{1}}
\newcommand{\tr}{\text{Tr}}

\newcommand{\comp}{\hbox{\hskip0.75mm$\circ\hskip-1.0mm\circ$\hskip0.75mm}}

\newtheorem{thm}{Theorem}
\newtheorem{proposition}{Proposition}
\newtheorem{lemma}{Lemma}

\newtheorem{definition}{Definition}

\newtheorem{remark}{Remark}

\begin{document}

\title{Genuine certification of incompatible quantum instruments through sequential communication tasks}

\author{Arindam Mitra}
\email{am56@iitbbs.ac.in}
\email{arindammitra143@gmail.com}
\affiliation{Department of Physics, School of Basic Sciences, Indian Institute of Technology Bhubaneswar, Odisha 752050, India}
\author{Satyaki Manna}
\email{mannasatyaki@gmail.com}
\affiliation{Department of Physics, School of Basic Sciences, Indian Institute of Technology Bhubaneswar, Odisha 752050, India}
\author{Debashis Saha}
\email{dsaha@iitbbs.ac.in}
\affiliation{Department of Physics, School of Basic Sciences, Indian Institute of Technology Bhubaneswar, Odisha 752050, India}


\begin{abstract}
Quantum instruments constitute the general description of quantum dynamics, encompassing both quantum measurements and quantum channels as special cases. Consequently, the incompatibility of quantum instruments represents a fundamental manifestation of nonclassicality in quantum theory. Here, we establish the operational significance of this notion by demonstrating communication tasks with classical inputs and outputs that enable the semi-device-independent certification of incompatible quantum instruments. We introduce a class of three-party communication tasks involving a sender, a relayer, and a receiver, and derive the tight upper bounds of the figure of merits of these tasks achievable by all compatible instruments implemented by the relayer. Furthermore, we show that these bounds coincide with the optimal performance attainable in a classical communication subject to the same dimensional constraints. Violation of this bound certifies the incompatibility of the pair of quantum instruments implemented by the relayer. This identifies certification of incompatible instruments as a manifestation of quantum advantage in communication.  This certification protocol is \emph{genuine} as it is able to certify the incompatibility of a pair of instruments where the set of induced measurements and the set of induced channels are compatible and, therefore it \emph{does not depend} on the incompability of measurements and channels induced by the instruments. Finally, we identify the simplest instance of our communication scenario that enable the certification of incompatible quantum instruments.
\end{abstract}

\maketitle



\section{Introduction}
Quantum incompatibility is one of the defining features that distinguishes quantum theory from its classical counterpart \cite{Heinosaari_incompatibility_review_2016,Uola_incomp_meas_review_2023}. It is defined as the impossibility of simultaneously implementing a set of devices on a single input state, where the devices may be measurements \cite{Heinosaari_incompatibility_review_2016,Uola_incomp_meas_review_2023}, channels \cite{Heinosaari_incompatibility_review_2016,haapasalo2015_incomp_chan,Heinosaari_incomp_channel_2017,Mori_incomp_chan_2020,girard2021jordan,Mitra_incomp_chan_2024,Miyadera_incomp_chan_2024,carmeli2019witnessing,Nechita2022fisher,guo2026irreversibility,d2022incompatibility}, instruments \cite{Mitra_comp_in_2022,Mitra_in_2023,Leppajarvi_incomp_in_2024,Ghai_in_comp_2025,d2022incompatibility,heinosaari2014stro_incomp,Chitambar_incomp_in_2024,Uola_incomp_in_2020,lever2016measure_msc_thesis,gudder2020finite,Hsieh_incomp_in_2024,sau_demul_2026,Buscemi2023unifyingdifferent}, testers \cite{Sedlak_incomp_tester_2016}, and other quantum devices. Although incompatibility may initially appear to be a limitation of the theory, it is now understood as a fundamental resource for several well-known information-theoretic tasks \cite{Heinosaari_incompatibility_review_2016,Uola_incomp_meas_review_2023,Buscemi_res_incomp_2020} and for the demonstration of several non-classicality \cite{Heinosaari_incompatibility_review_2016,Uola_incomp_meas_review_2023}. These include the demonstration of Bell nonlocality \cite{fine,wolf09} and quantum steering, and quantum advantage in some quantum state discrimination tasks \cite{Skrzypczyk19,Mori_incomp_chan_2020}, some quantum state exclusion tasks \cite{Uola_incomp_state_exclu}, quantum random access code \cite{Carmeli_2020,saha20} and other communication tasks \cite{saha23}. Although measurement incompatibility and channel incompatibility have been widely explored for a long time, instrument incompatibility has attracted significant attention only recently. In Refs. \cite{Mitra_comp_in_2022,Leppajarvi_incomp_in_2024}, it has been argued that parallel compatibility of a given pair of instruments can capture the notion of incompatibility of measurements, and channels induced by those instruments, and hence it is conceptually complete, unlike the other inequivalent notion of instrument compatibility, namely, traditional compatibility \cite{Mitra_comp_in_2022,heinosaari2014stro_incomp,Chitambar_incomp_in_2024,Uola_incomp_in_2020,lever2016measure_msc_thesis,gudder2020finite,Hsieh_incomp_in_2024,sau_demul_2026} which cannot capture the well-established notion of channel compatibility \cite{Heinosaari_incompatibility_review_2016,haapasalo2015_incomp_chan,Heinosaari_incomp_channel_2017,Mori_incomp_chan_2020,girard2021jordan,Mitra_incomp_chan_2024,Miyadera_incomp_chan_2024,carmeli2019witnessing,Nechita2022fisher,guo2026irreversibility,d2022incompatibility}. Therefore, in this work, we restrict ourselves to parallel (in)compatibility.

Over the years, considerable effort has been devoted to understanding the operational power of incompatible measurements through device-independent or semi-device-independent witnesses that depend purely on observed input-output statistics. A prominent line of research concerns the certification of measurement incompatibility, which can be inferred from the violation of Bell inequalities \cite{banik13,wolf09,loulidi,Zhu_2024} or demonstrated through Einstein–Podolsky–Rosen (EPR) steering \cite{uola14,quintino,sarkar22,chen16,banik15,bavaresco}. As mentioned above, beyond Bell-type scenarios, incompatible measurements have also been identified as a resource for achieving quantum advantages in prepare-and-measure communication tasks, thereby providing alternative methods for certifying incompatibility \cite{saha23,Skrzypczyk19,carmeli18,carmeli19,Guerini19,joannou,das26,Designolle_2019,uola21,Quintino19,degois, Vieira_2023}. In contrast, the certification of incompatibility for more general quantum devices, such as quantum channels and quantum instruments, remains unexplored.

As measurements and channels can be considered as special cases of quantum instruments, the certification of the incompatibility of instruments is of immense importance. Therefore, it is a matter of investigation whether one can design a class of instrument incompatibility certification protocols that encompass the certification of both measurement and channel incompatibility as special cases. Such a unified approach will be important not only from the foundational point of view but also for applications. Importantly, to be a \emph{genuine} certification protocol, it is better if the certification protocol \emph{does not depend} on the certification of the incompatibility of measurements and channels induced by the instruments. In this work, we are \emph{motivated} to design and study such a class of certification protocols.

The rest of the paper is organized as follows. In Sec. \ref{sec:prelim}, we discuss the preliminaries. In Sec. \ref{Sec:seq_comm_task}, and Sec. \ref{sec:certific_in}, we report the main results. More specifically, in Sec. \ref{Sec:seq_comm_task}, we report some generic results for   sequential communication tasks described in Fig. \ref{fig:communication_scenario}(a). In Sec.  \ref{sec:certific_in}, after explaining the $(d_A,d)$-sequential XOR task, we show the genuine certification of instrument incompatibility through this task for $d_A=d=2$. Then, we generalize our results for $d_A=d=d$ where any integer $d\geq 2$ and lastly, we discuss the genuine certification of the same in the minimal scenario where Alice has only three inputs. Finally, in Sec. \ref{sec:conc}, we summarize our results and discuss future directions. 

\section{Quantum measurements, channels and instruments and their compatibility}\label{sec:prelim}
A quantum measurement $M$ acting on $\cH$ is mathematically described by a set of positive semi-definite matrices i.e., $M:=\{M(a)\geq 0,M\in\cL(\cH)\}^{n}_{a=1}$ such that $\sum^{n}_{a=1}M(a)=\Id_{\cH}$ where $\cL(\cH)$ is the set of linear operators on Hilbert space $\cH$ and $\Id_{\cH}$ is the identity operator on Hilbert space $\cH$ \cite{Heinosaari_Ziman_book,Hayashi_book}. Measurements are also known as POVMs. The set $\Omega_{M}=\{1,\ldots,n\}$ is known as the outcome set of $M$ and each $M(a)$ is known as the POVM elements of $M$. We denote the set of all measurements on Hilbert space $\cH$ as $\mathscr{M}(\cH)$ and the set of all quantum states on Hilbert space $\cH$ as $\cS(\cH)$. If the measurement $M$ is implemented on a quantum system with an arbitrary quantum state $\rho\in\cS(\cH)$, the probability of an arbitrary outcome $a\in\Omega_M$ is given by $p(a|\rho,M)=\tr[\rho M(a)]$. Outcomes of an arbitrary measurement can be post processed. A measurement $M^{\prime}\in\mathscr{M}(\cH)$ is said to a post processing of $M\in\mathscr{M}(\cH)$ if for all $b\in\Omega_{M^{\prime}}$, $M^{\prime}(b)=\sum_{a\in\Omega_M}\nu_{ba} M(a)$ where $\nu_{ba}\geq 0\forall a,b$ and $\sum_{b\in\Omega_{M^{\prime}}}\nu_{ba}=1$. We denote it as $M^{\prime}\preceq M$. A pair of measurements $(M_1,M_2)\subset\mathscr{M}(\cH)$ is said to be compatible or jointly implementable if there exists a joint measurement $M\in\mathscr{M}(\cH)$ with outcome set $\Omega_{M}=\Omega_{M_1}\times\Omega_{M_2}$ such that for all $a\in\Omega_{M_1}$, $b\in\Omega_{M_2}$, the equations $M_1(a)=\sum_{q\in\Omega_{M_2}}M(a,q)$ and  $M_2(b)=\sum_{r\in\Omega_{M_1}}M(r,b)$ hold \cite{Heinosaari_incompatibility_review_2016,Uola_incomp_meas_review_2023}. Therefore, implementation of $M$ on a quantum system is equivalent to the simultaneous implementation of $M_1$ and $M_2$ on it. We denote it as $M_1\comp M_2$. It is known that if $M^{\prime}_1\preceq M_1$ and $M^{\prime}_2\preceq M_2$ then $M_1\comp M_2\Rightarrow M^{\prime}_1\comp M^{\prime}_2$.

A quantum channel $\Lambda:\cL(\cH)\rightarrow\cL(\cK)$ is mathematically defined as a completely positive trace preserving (CPTP) linear map \cite{Heinosaari_Ziman_book,Hayashi_book}. We denote the set of all quantum channels with input Hilbert space $\cH$ and output Hilbert space $\cK$ as $\mathscr{C}(\cH,\cK)$. Acting on a quantum state $\rho\in\cS(\cH)$, it provides an output state $\Lambda(\rho)\in\cS(\cK)$. We denote the identity channel acting on $\cL(\cH)$ as $\mathbbm{I}_{\cH}$. The action of a CP linear map $\Phi$ in Heisenberg picture, denoted as $\Phi^{\dagger}:\cL(\cH)\rightarrow\cL(\cK)$ is defined by the equation
$    \tr[\Lambda(X)Y]=\tr[X\Lambda^{\dagger}(Y)]~\forall X\in\cL(\cH), Y\in\cL(\cK).$ The linear map $\Phi^{\dagger}:\cL(\cK)\rightarrow\cL(\cH)$ is CP as $\Phi$ is CP. In Heisenberg picture, a quantum channel $\Lambda\in\mathscr{C}(\cH,\cK)$ acts on a measurement $M\in\mathscr{M}(\cK)$ as $\Lambda^{\dagger}[M]:=\{\Lambda^{\dagger}(M(a))\}\in\mathscr{M}(\cH)$. A pair of quantum channels $\Lambda_1\in\mathscr{C}(\cH,\cK_1)$ and $\Lambda_2\in\mathscr{C}(\cH,\cK_2)$ is said to be compatible if there exists a joint quantum channel $\Lambda:\mathscr{C}(\cH,\cK_1\otimes\cK_2)$ such that for an arbitrary quantum state $\rho\in\cS(\cH)$, we have $\Lambda_1(\rho)=\tr_{\cK_2}\Lambda(\rho)$ and $\Lambda_2(\rho)=\tr_{\cK_1}\Lambda(\rho)$ \cite{Heinosaari_incompatibility_review_2016,Heinosaari_incomp_channel_2017}. In short-hand notation, we denote these as $\Lambda_1=\tr_{\cK_2}\circ\Lambda$ and $\Lambda_2=\tr_{\cK_1}\circ\Lambda$. Such short-hand notation will be used throughout the paper.  The implementation of $\Lambda$ on a quantum system is equivalent to the simultaneous implementation of $\Lambda_1$ and $\Lambda_2$ on it. We denote it as $\Lambda_1\comp\Lambda_2$.

 A quantum instrument $\cI$ is defined as a set of CP trace-non increasing maps i.e., $\cI:=\{\Phi_b:\cL(\cH)\rightarrow\cL(\cK)\}^n_{b=1}$ such that $\sum^n_{b=1}\Phi_b=\Phi$ is a quantum channel \cite{Heinosaari_Ziman_book,Hayashi_book}. The outcome set of the instrument $\cI$ is $\Omega_{\cI}:=\{1,\ldots,n\}$. Each $b$ is known as the classical output of $\cI$ and given a quantum state $\rho\in\cS(\cH)$, it occurs with the probability $p(b|\cI,\rho)=\tr[\Phi_b(\rho)]$. The quantum output $\cI$, given $b$ is already obtained, is $\frac{\Phi_b(\rho)}{\tr[\Phi_b(\rho)]}$. Besides inducing the channel $\Phi$, instrument $\cI$ also induces a measurement $\hat{M}_{\cI}:=\{\hat{M}_{\cI}(b):=\Phi^{\dagger}_b(\Id_{\cK})\}$. We denote the set of all quantum instruments with input Hilbert space $\cH$ and output Hilbert space $\cK$ as $\mathscr{I}(\cH,\cK)$. It can be easily seen that quantum measurements and channels can be considered special cases of quantum instruments. A quantum instrument $\cI=\{\Phi_b\}\in\mathscr{I}(\cH,\cK)$ acts on a measurement $M\in\mathscr{M}(\cK)$ and provides the measurement $\cI^{\dagger}[M]:=\{\cI^{\dagger}[M](b,a):=\Phi^{\dagger}_b(M(a))\}\in\mathscr{M}(\cH)$. Clearly, $\Omega_{\cI^{\dagger}[M]}=\Omega_{\cI}\times\Omega_{M}$. A pair of instruments $\cI_1=\{\Phi^1_a\}\in\mathscr{I}(\cH,\cK_1)$ and $\cI_2=\{\Phi^2_b\}\in\mathscr{I}(\cH,\cK_2)$ is said to be (parallel) compatible if there exists a joint instrument $\cI=\{\Phi_{ab}\}\in\mathscr{I}(\cH,\cK_1\otimes\cK_2)$ with outcome set $\Omega_{\cI_1}\times\Omega_{\cI_2}$ such that the following holds:
$\Phi^1_a=\sum_{b\in\Omega_{\cI_2}}\tr_{\cK_2}\circ\Phi_{ab}~\forall a\in\Omega_{\cI_1},~\Phi^2_a=\sum_{a\in\Omega_{\cI_1}}\tr_{\cK_1}\circ\Phi_{ab}~\forall b\in\Omega_{\cI_2}.$ \cite{Mitra_comp_in_2022,Mitra_in_2023,Leppajarvi_incomp_in_2024,Ghai_in_comp_2025}.
 We denote it as $\cI_1\comp\cI_2$. Clearly, the implementation of $\cI$ is equivalent to the simultaneous implementation of $\cI_1$ and $\cI_2$. If a pair of instruments is compatible, their induced measurements, are compatible and induced channels are compatible \cite{Mitra_comp_in_2022}. In this work, we highlight that the converse is not true. It should be mentioned here that the notion of measurement-channel compatibility also exists in the literature. A measurement $M\in\mathscr{M}(\cH)$ and a channel $\Lambda\in\mathscr{C}(\cH,\cK)$ are said to be compatible if there exists a quantum instruments $\cI=\{\Phi_b\}\in\mathscr{I}(\cH,\cK)$ such that $M=\{\Phi^{\dagger}_b(\Id_{\cK})\}_{b\in\Omega_{\cI}=\Omega_{M}}$ and $\Lambda=\sum_b\Phi_b$ \cite{Heinosaari_incompatibility_review_2016,Mori_incomp_chan_2020,Heinosaari_meas_chan_comp_2018}. We denote it as $M\comp\Lambda$. 
\section{Sequential Communication Task}
\label{Sec:seq_comm_task}
First, let us establish the notation $[K]$ to represent the set $\{0, \cdots, K-1\}$ for any positive integer $K$.
 We define a sequential communication task involving three parties: Alice (the sender), Bob (the relayer), and Charlie (the receiver). In this task, Alice and Bob receive inputs \(x \in [n_A]\) and \(y \in [n_B]\), respectively. Upon receiving the input \(x\), Alice sends a \(d_A\)-dimensional classical or quantum system to Bob. After receiving both the input \(y\) and the message from Alice, Bob produces an output \(b_y=f(x,y)\in[d_y]\) and subsequently sends a \(d_B\)-dimensional classical or quantum system to Charlie. Charlie does not receive any input; based solely on the message received from Bob, Charlie produces an output \(c\in[d_c]\). The performance of the communication task is characterized by the set of conditional probability distributions \(\{p(b_y,c|x,y)\}\). Any figure of merit can be written as the function of these probabilities, which means, the figure of merit 
 \bea\label{fig_merit}
 \cS=\sum_{x,y,b_y,c}f[\{p(b_y,c|x,y)\}].
 \eea
In classical communication, the classical messages transmitted by Alice and Bob are denoted by $m_A$ and $m_B$ respectively. Any typical probability in this case can be expressed as
\bea\label{c_stat}
p(b_y,c|x,y)=\sum_{m_A,m_B} p_e(m_A|x)p_t(b_y,m_B|m_A,y)p_d(c|m_B).\label{Eq:prob_class_strat}
\eea
The individual encoding, transformation and decoding probability satisfies non-negativity and normalization conditions which reads as,
$\sum_{m_A} p_e(m_A|x)=\sum_{b_y,m_B} p_t(b_y,m_B|m_A,y)= \sum_{c}p_d(c|m_B)=1$. 

In quantum communication, Alice encodes her input into a quantum state $\rho_x\in\cS(\cH_A)$ where $\dim(\cH_A)=d_A$ and sends it to Bob. Depending on the input $y$, Bob implements an instrument from the set $\{\cI_y\}_y$ on it. Bob produces a classical output $b_y$ with probability $\tr[\Phi^y_{b_y}(\rho_x)]$ and transmits quantum output $\frac{\Phi^y_{b_y}(\rho_x)}{\tr[\Phi^y_{b_y}(\rho_x)]}\in\cS(\cH_B)$ to Charlie where $\dim(\cH_B)=d_B$. Having no input, Charlie performs a fixed measurement $N:=\{N(c)\}_{c=1}^{d_c}$ on the state sent by Bob. The expression of the probabilities reduces to, 
\bea\label{q_stat}
p(b_y,c|x,y)=\tr\left(\Phi^y_{b_y}(\rho_x)N(c)\right).
\eea
We denote this sequential communication scenario by the set of natural numbers $(n_A,d_A,n_B,d_y,d_B,d_c)$.  Furthermore, both $\cS^\cC_{op}$ and $\cS^\cQ_{op}$ depend on these parameters.  However, for notational simplicity, we sometimes omit the explicit dependence of $\cS^\cC_{op}$ and $\cS^\cQ_{op}$ on these parameters when it is obvious from the context. For a given quantum strategy in a given task, whenever the figure of merit $\cS^\cQ>\cS^\cC_{op}$, we interpret this as a \emph{quantum advantage} by that quantum strategy in that task.

Now, we define the set of all classical statistics which is generated according to Eq. \eqref{c_stat} by 
\bea
\mathcal{C}(n_A,d_A,n_B,d_y,d_B,d_c):=\{p(b_y,c|x,y)\}.
\eea
Similarly, set of all probabilities generated by quantum communication (i.e., from Eq. \eqref{q_stat}) is given by 
\bea
\cQ(n_A,d_A,n_B,d_y,d_B,d_c):=\{p(b_y,c|x,y)\}.
\eea
 We can define another type of statistics generated in quantum communication when the instruments $\{\cI_y\}_y$ is compatible i.e, when $\{\cI_y\}_y$ is compatible  in Eq. \eqref{q_stat}. We call this
 \bea
\cQ^C(n_A,d_A,n_B,d_y,d_B,d_c):=\{p(b_y,c|x,y)\}.
\eea
 Once the scenario $(n_A,d_A,n_B,d_y,d_B,d_c)$ is fixed, we write $\mathcal{C}$ instead of $\mathcal{C}(n_A,d_A,n_B,d_y,d_B,d_c)$ for notational simplicity. Such notations will also be useed for $\cQ$ and $\cQ^C$  As compatible instruments form a subset of all quantum instruments, for any scenario $(n_A,d_A,n_B,d_y,d_B,d_c)$ we can infer, $\cQ^C\subseteq \cQ$. Furthermore, Trivially, we know, $\cC\subseteq \cQ$.
\
\begin{figure}[hbt!]
    \centering
    \includegraphics[scale=0.29]{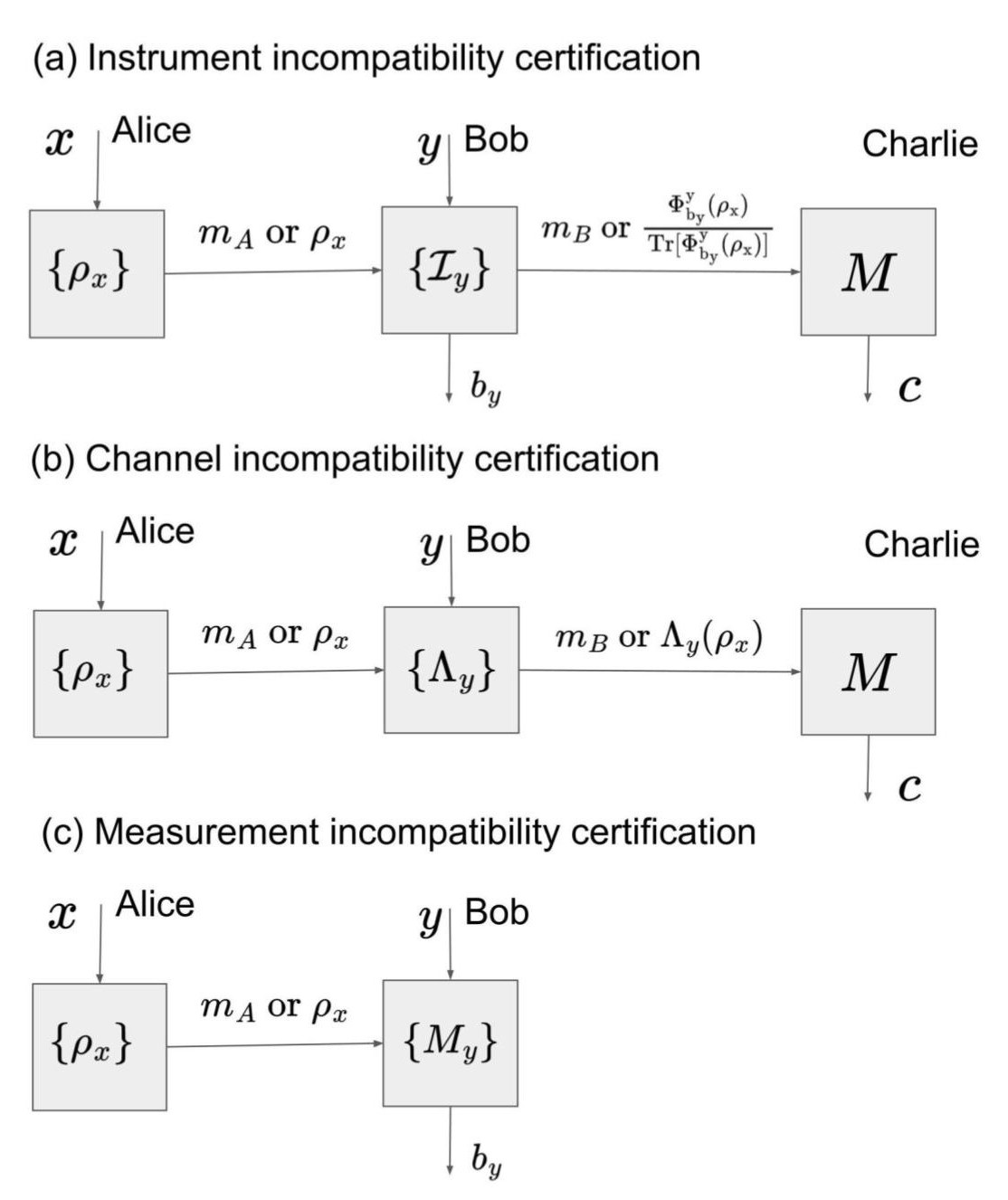}
    \caption{ Fig. (a) depicts the scenario of instrument incompatibility certification that we explore in this work. The scenarios depicted Fig. (b) and Fig. (c) are the scenarios of measurement and channel incompatibility certification scenarios respectively and are the special cases of the scenario depicted in Fig. (a) which is unified. Additionally, note that as Bob cannot send Charlie both $y$, and $b_y$ due to the dimensional restriction, in general, the scenario depicted in Fig. (a) is not equivalent to measure-and-prepare scenario, in general.}
    \label{fig:communication_scenario}
\end{figure}
Clearly, if the shared randomness is allowed, the set of statistics obtained from classical strategy, quantum strategy, and strategy where Bob uses compatible instruments become $\conv \mathcal{C}(n_A,d_A,n_B,d_y,d_B,d_c)$, $\conv \cQ(n_A,d_A,n_B,d_y,d_B,d_c)$, and $\conv \cQ^C(n_A,d_A,n_B,d_y,d_B,d_c)$ repectively where for any set $A$, $\conv A$ denotes the convex hull of the set.
Now, for any scenario $(n_A,d_A,n_B,d_y,d_B,d_c)$,  the optimum classical figure of merit emerges from Eq. \eqref{fig_merit} as, 
\bea
\cS^{\cC}_{op}&=&\max_{\conv \mathcal{C}}f[\{p(b_y,c|x,y)\}]\nonumber\\
\eea
Similarly, the optimum quantum figure of merit is given by
\bea\label{S_q}
\cS^{\cQ}_{op}=\max_{\conv \cQ}f[\{p(b_y,c|x,y)\}].
\eea
and the optimum figure of merit achieved by compatible quantum instruments is given by
\bea\label{S_qc}
\cS^{\cQ^C}_{op}=\max_{\conv \cQ^C} f[\{p(b_y,c|x,y)\}].
\eea
 If $f$ is a linear figure of merit then $f[\{p(b_y,c|x,y)\}=\sum_{x,y,b_y,c} \alpha(x,y,b_y,c)\tr\left(\Phi^y_{b_y}(\rho_x)N(c)\right)$  where $\alpha(x,y,b_y,c)$ is a real number. Strategies with shared randomness form the convex hull of those without it, and a linear figure of merit always attains its maximum over a convex hull on the original set, so shared randomness offers no advantage. Therefore, in Sec.~\ref{sec:certific_in} we derive all bounds without it.

Now, we want to compare the set of classical statistics with the set of statistics coming from quantum communication with compatible instruments. For this reason, we state the following result of Frenkel and Weiner~\cite{fw} that we need to use.

\begin{lemma}[Frenkel--Weiner~\cite{fw}]\label{lem:FW}
Let $\{\rho_x\}_{x\in X}\subset\mathcal{S}(\mathbb{C}^{d})$ be a finite set of
states and $\{E(o)\}_{o\in O}$ a POVM on $\mathbb{C}^{d}$ with finitely many
outcomes. Then there exist a finite set $\Lambda$, a probability distribution
$\{\pi(\lambda)\}_{\lambda\in\Lambda}$ and, for every $\lambda\in\Lambda$,
an encoding $\{p_e(m|x,\lambda)\}$ with $m\in[d]$ and a decoding
$\{p_d(o|m,\lambda)\}$, such that $\forall\, x, o$,
\begin{equation}\label{eq:FW}
  \tr\big[\rho_x E(o)\big]
  =\sum_{\lambda\in\Lambda}\pi(\lambda)\sum_{m\in[d]}
  p_e(m|x,\lambda)\,p_d(o|m,\lambda).
\end{equation}
Moreover, the encodings and decodings can be chosen deterministic, i.e.,
there exist functions $e_\lambda:X\to[d]$ and $g_\lambda:[d]\to O$ with
\begin{equation}\label{eq:FW_det}
  p_e(m|x,\lambda)=\delta_{m,\,e_\lambda(x)},\qquad
  p_d(o|m,\lambda)=\delta_{o,\,g_\lambda(m)} .
\end{equation}
\end{lemma}
Now, we are ready to state and prove the following result.

{\begin{proposition}\label{cor1}
For any sequential communication task denoted by a set of natural numbers $(n_A,d_A,n_B,d_y,d_B,d_c)$, the following statements hold.
\begin{enumerate}[label=(\roman*)]
    \item $\cC\subseteq \cQ^C$. Consequently, for any figure of merit $\cS$, we have $\cS^\cC_{op}\leq \cS^{\cQ^C}_{op}.$
    \item If either the condition (a) \quad $d_B\ \geq\ \min\{d_c,\ n_B d_A\}$, or the condition (b)\quad $n_A\ \leq\ d_A$ or both are satisfied then we have
    $\QCset \subseteq\conv\Cset$ and $\conv\Cset=\conv\QCset$, and consequently $\mathcal{S}^{C}_{op}=\mathcal{S}^{Q^C}_{op}$ for every
figure of merit $\mathcal{S}$ of the form~(1).
\end{enumerate}
\end{proposition}

\begin{proof}
     \textit{Proof of statement (i).} To show statement (i), we have to prove that any probability distribution $\{p(b_y,c|x,y)\}$ obtained from any classical strategy (given in Eq.~\eqref{Eq:prob_class_strat}) can be obtained using quantum strategy where Bob uses compatible quantum instruments. For notational simplicity, we prove it for $n_B=2$. The generalization of the results for arbitrary $n_B$ is straightforward. 

    Consider the orthonormal basis $\{\ket{m_A}\}$ of the Hilbert space $d_A$-dimensional Hilbert space $\cH_A$ and $\{\ket{m_B}\}$ of the Hilbert space $d_B$-dimensional Hilbert space $\cH_B$. Assume that Alice encodes her input $x$ as $\rho_x=\sum_{m_A}p_e(m_A|x)\ket{m_A}\bra{m_A}$, Bob uses a pair of quantum instruments $\cI_1$, and $\cI_2$ such that $\cI_{y}=\{\Phi^y_{b_y}\}$, where $\Phi^{y}_{b_y}(\rho)=\sum_{m_A,m_B}p_t(b_y,m_B|m_A,y)\tr[\rho\ket{m_A}\bra{m_A}]\ket{m_B}\bra{m_B}$ for an arbitrary $\rho\in\cS(\cH_A)$ and Charlie uses a measurement $M=\{M(c)=\sum_{m_B}p_d(c|m_B)\ket{m_B}\bra{m_B}\}\in\mathscr{M}(\cH_B)$. Then
    \begin{align}
        p(b_y,c|x,y)=&\tr[\Phi^y_{b_y}(\rho_x)M(c)]\nonumber\\
        =&\sum_{m_A,m_B} p_e(m_A|x)p_t(b_y,m_B|m_A,y)p_d(c|m_B).
    \end{align}

    Now, it is easy to see that $\cI_1$ and $\cI_2$ are compatible with the joint instrument $\cI=\{\Phi_{b_0b_1}\}$ where for an arbitrary $\rho\in\cS(\cH_A)$
    
    \begin{align}
        \Phi_{b_0b_1}(\rho)=&\sum_{m_A,m_B,m^{\prime}_B}p(b_0,m_B|m_A,0)p(b_1,m^{\prime}_{B^{\prime}}|m_A,1)\nonumber\\
        &\tr[\rho\ket{m_A}\bra{m_A}]  \ket{m_B}\bra{m_B}\otimes\ket{m^{\prime}_{B^{\prime}}}\bra{m^{\prime}_{B^{\prime}}},
    \end{align}
    where for all $b_0\in\{0,\ldots,d_0-1\}$ and $b_1\in\{0,\ldots,d_1-1\}$, $\{\ket{m^{\prime}_{B^{\prime}}}\}$ is a basis of Hilbert space $\cH_{B^{\prime}}$ where $\cH_B=\cH_{B^{\prime}}$ and $\ket{m_{B}}=\ket{m_{B^{\prime}}}$. Hence, $\Cset\subseteq\QCset$ which implies
$\conv\Cset\subseteq\conv\QCset$. Therefore, we have $\cS^\cC_{op}\leq \cS^{\cQ^C}_{op}.$

\textit{Proof of statement (ii).}
Recall that $\conv\Cset$ consists of all statistics of the form
\begin{align}\label{eq:convC_form}
  p(b,c|x,y)=\sum_{\lambda\in\Lambda}\pi(\lambda)
\sum_{m_A,m_B}p_e(m_A|x,\lambda)&p_t(b,m_B|m_A,y,\lambda)\nonumber\\
  &\times p_d(c|m_B,\lambda),
\end{align}
where $m_A\in[d_A]$, $m_B\in[d_B]$, $\{\pi(\lambda)\}$ is a probability
distribution of the shared randomness, and, for every fixed $\lambda$, $p_e$, $p_t$ and $p_d$ are the encoding, transformation and decoding. Hence,
to prove $\QCset\subseteq\conv\Cset$, it suffices to show that every
$p\in\QCset$ can be written in the form~\eqref{eq:convC_form}.

Let $p\in\QCset$ be generated by states $\{\rho_x\}_{x\in[n_A]}\subset
\mathcal{S}(\cH_A)$, a compatible family of instruments
$\{\cI_y=\{\Phi^y_{b}\}_{b\in[d_y]}\}_{y\in[n_B]}$ and Charlie's POVM
$N=\{N(c)\}_{c\in[d_c]}$ on $\cH_B$, so that
$p(b,c|x,y)=\tr\big[\Phi^y_b(\rho_x)N(c)\big]$.

\emph{Case (a): $d_B\geq\min\{d_c,n_Bd_A\}$.}
Since the instruments are
compatible, there exists a joint instrument $\cI=\{\Phi_{\mathbf b}\}$, with
$\mathbf b=(b_0,\dots,b_{n_B-1})\in[d_y]^{n_B}$, from $\cH_A$ to $\overline{\cH}_B=\bigotimes_{i=0}^{n_B-1}\cH_B^{i}$, where $\cH_B^{i}=\cH_B$ for all $i\in[n_B]$, such that
$\sum_{\mathbf b:\,b_y=b}\tr_{\neg y}\circ\Phi_{\mathbf b}=\Phi^y_b$ for all
$y$ and $b$. Here $\tr_{\neg y}$ denotes the partial trace over all $\cH_B^{i}$ with $i\neq y$. For $\mathbf c=(c_0,\dots,c_{n_B-1})\in[d_c]^{n_B}$ define
\begin{equation}
  G(\mathbf b,\mathbf c):=\Phi_{\mathbf b}^{\dagger}
  \big(N(c_0)\otimes\cdots\otimes N(c_{n_B-1})\big).
\end{equation}
Each $G(\mathbf b,\mathbf c)$ is positive semidefinite because
$\Phi^\dagger_{\mathbf b}$ is completely positive, and
$\sum_{\mathbf b,\mathbf c}G(\mathbf b,\mathbf c)
=\sum_{\mathbf b}\Phi^\dagger_{\mathbf b}(\mathbbm{1}_{\overline{\cH}_B})=\mathbbm{1}_{\cH_A}$
because $\sum_{\mathbf b}\Phi_{\mathbf b}$ is trace preserving. Hence $G$ is a POVM on the $d_A$-dimensional space $\cH_A$. Summing over all components of $\mathbf b$ and $\mathbf c$ except the $y$-th, and using $\sum_{c'}N(c')=\mathbbm{1}_{\cH_B}$, we obtain
\begin{equation}\label{eq:marginal}
  p(b,c|x,y)=\tr\big[\Phi^y_b(\rho_x)N(c)\big]
  =\sum_{(\mathbf b,\mathbf c):\,(b_y,c_y)=(b,c)}
  \tr\big[\rho_x\,G(\mathbf b,\mathbf c)\big].
\end{equation}
Applying
Lemma~\ref{lem:FW} in its deterministic form~\eqref{eq:FW_det}, with
$n=d_A$, the states $\{\rho_x\}$ and the POVM $\{G(\mathbf b,\mathbf c)\}$, we obtain
\begin{equation}
  \tr\big[\rho_x\,G(\mathbf b,\mathbf c)\big]
  =\sum_{\lambda}\pi(\lambda)\sum_{m_A\in[d_A]}\delta_{m_A,\,e_\lambda(x)}\,
  \delta_{(\mathbf b,\mathbf c),\,g_\lambda(m_A)}.
  \end{equation}
Thus, for every $\lambda$, each classical
message $m_A$ prescribes one output pair $(b^\lambda_y(m_A),c^\lambda_y(m_A))$
for every input $y$ of Bob. Substituting into Eq.~\eqref{eq:marginal} gives
\begin{equation}\label{eq:caseA_target}
  p(b,c|x,y)=\sum_{\lambda}\pi(\lambda)\sum_{m_A\in[d_A]}
  \delta_{m_A,\,e_\lambda(x)}\,\delta_{b,\,b^\lambda_y(m_A)}\,
  \delta_{c,\,c^\lambda_y(m_A)} .
\end{equation}
Fix
$\lambda$ and let
$C_\lambda:=\{\,c^{\lambda}_y(m_A)\;:\;m_A\in[d_A],\ y\in[n_B]\,\}
  \subseteq[d_c]$
be the set of all values of $c$ that Charlie is ever required to output for
this $\lambda$. It contains at most $n_Bd_A$ elements and at most $d_c$ elements, so condition (a) gives
$|C_\lambda|\leq\min\{d_c,n_Bd_A\}\leq d_B$. Under this condition, there exists an injective labelling $q_\lambda:C_\lambda\to[d_B]$, i.e., distinct values of $c$
are encoded in distinct labels by Bob. Thus, we define the strategy by Bob and Charlie as
\begin{align}
  p_t(b,m_B|m_A,y,\lambda)&:=\delta_{b,\,b^\lambda_y(m_A)}\;
  \delta_{m_B,\,q_\lambda(c^\lambda_y(m_A))},\label{eq:pt_caseA}\\
  p_d(c|m_B,\lambda)&:=
  \begin{cases}
    \delta_{c,\,q_\lambda^{-1}(m_B)}, & m_B\in q_\lambda(C_\lambda),\\
    \delta_{c,0}, & \text{otherwise}.
  \end{cases}\label{eq:pd_caseA}
\end{align}
Since $q_\lambda$ is injective,
$q_\lambda^{-1}\big(q_\lambda(c^\lambda_y(m_A))\big)=c^\lambda_y(m_A)$,
and therefore
\begin{equation}
  \sum_{m_B\in[d_B]}p_t(b,m_B|m_A,y,\lambda)\,p_d(c|m_B,\lambda)
  =\delta_{b,\,b^\lambda_y(m_A)}\,\delta_{c,\,c^\lambda_y(m_A)} .
\end{equation}
Inserting this into Eq.~\eqref{eq:caseA_target} yields \eqref{eq:convC_form}. Hence $p\in\conv\Cset$.

\emph{Case (b): $n_A\leq d_A$.} Since $n_A\leq d_A$, we may identify
$[n_A]$ with a subset of $[d_A]$ and let Alice send her input itself to Bob as
$p_e(m_A|x):=\delta_{m_A,x}$. 
Bob then knows $x$, and thus, the task essentially reduces to a prepare-and-measure task. Specifically, given any quantum strategy where
\begin{equation}\label{eq:caseB_factor}
  p(b,c|x,y)=\tr\big[\Phi^y_b(\rho_x)\big]\tr\big[\tau_{x,y,b}N(c)\big],
\end{equation}
with $\tau_{x,y,b}:=\frac{\Phi^y_b(\rho_x)}{q(b|x,y)}\in\mathcal{S}(\cH_B)$, and $m_A=x$, let us consider $p_t(b,m_B|m_A,y,\lambda)$ and $p_d(c|m_B,\lambda)$ such that
\begin{equation}\label{eq:pt-(b)}
p_t(b,m_B|m_A,y,\lambda):=\bar{p}(b|m_A,y,\lambda)
  \tilde p_e(m_B|m_A,y,b,\lambda), 
\end{equation}
where $\bar{p}(b|m_A,y,\lambda) = \tr\big[\Phi^y_b(\rho_x)\big]$ and 
\begin{equation}\label{eq:FW_caseB}
  \tr\big[\tau_{x,y,b}N(c)\big]
  =\sum_{\lambda}\pi(\lambda)\sum_{m_B\in[d_B]}
  \tilde p_e(m_B|x,y,b,\lambda)\,p_d(c|m_B,\lambda),
\end{equation}
which follows from Lemma~\ref{lem:FW}.
Substituting \eqref{eq:pt-(b)}-\eqref{eq:FW_caseB} into \eqref{eq:caseB_factor}, along with $p_e(m_A|x)=\delta_{m_A,x}$, yields \eqref{eq:convC_form}. 

In both cases $\QCset\subseteq\conv\Cset$. Taking convex
hulls gives $\conv\QCset\subseteq\conv\Cset$, while statement (i) gives
$\conv\Cset\subseteq\conv\QCset$. Therefore $\conv\Cset=\conv\QCset$, and
since $\mathcal{S}^{\cC}_{op}$ and $\mathcal{S}^{\cQ^C}_{op}$ are defined as
maxima over $\conv\Cset$ and $\conv\QCset$, respectively,
$\mathcal{S}^{\cC}_{op}=\mathcal{S}^{\cQ^C}_{op}$.
\end{proof}

  In some scenarios, the inclusion relations $\cC\subseteq \cQ^C$ and $\conv\cC\subseteq \conv\cQ^C$ are strict as shown in Appendix \ref{app:compar_stat}.
  It is clear that in the three party sequential communication task that we are considering in this work, Charlie has no input (see Fig. \ref{fig:communication_scenario}). But similar scenarios where Charlie has inputs have been studied in Refs. \cite{Miklin_self_test_2020,Mohan_self_test_2019,Bowles_test_dim_2015,Hameedi_EA_vs_QD_2017,Roy_R_cert_2026,Sasmal_share_adv_2026,Mukherjee21,SS23} in different contexts. The reason for considering no input scenario for Charlie is that if Charlie is allowed to have input, the quantum advantage can be obtained even if Bob implements compatible instruments in cases where Charlie implements incompatible measurements. Consequently, a quantum advantage in such a scenario would no longer be able certify the incompatibility of Bob's quantum instruments. We refer the readers to Appendix \ref{App_sec:Charlie}

\section{Certification of incompatible instruments}
\label{sec:certific_in}

\begin{definition}
    An instrument incompatibility certification protocol is said to be genuine if it is able to certify a set of incompatible instruments even when the set of induced measurements is compatible and the set of induced channels is compatible.
\end{definition}

 In other words, a genuine instrument incompatibility \emph{should not depend} on the incompatibility of the measurements and channels induced by the instruments. Now, it is easy to see that if the figure of merit depends only on  $\{p(c|x,y)\}$ with the condition $d_B\geq d_c$ or $\{p(b_y|x,y)\}$, there is no advantage when the set of induced measurements or the set of induced channels are compatible respectively.  Therefore, it is better if the figure of merit of a certification protocol involves $\{p(b_y,c|x,y)\}$ non-trivially in order to make it genuine. We discuss it in Appendix \ref{App_genuine_cert} in detail. Furthermore, in Sec \ref{Subsec:2,2_XOR_task}, Sec. \ref{Subsec:d,d_XOR_task}, and Sec \ref{Subsec:minimal_scenario} as we consider a communication task where $d_B=d_c$, we are forced to consider a figure of merit that involves $\{p(b_y,c|x,y)\}$ non-trivially.

Before going into our main claims, it is useful to discuss the following Proposition. 
\begin{proposition}
    If a pair of instruments $\cI_1=\{\Phi^1_x\}\in\mathscr{I}(\cH,\cK_1)$ and $\cI_2=\{\Phi^2_y\}\in\mathscr{I}(\cH,\cK_2)$ is compatible (i.e., $\cI_1\comp\cI_2$), then for an arbitrary pair of measurements $M_1\in\mathscr{M}(\cK_1)$ and $M_2\in\mathscr{M}(\cK_2)$, we have $\cI_1^{\dagger}[M_1]\comp\cI_2^{\dagger}[M_2]$. \label{Prop:ins_comp_channel_comp}
\end{proposition}

\begin{proof}
    As $\cI_1\comp\cI_2$, we have a quantum instrument $\cI=\{\Phi_{xy}\}_{(x,y)\in\Omega_{\cI_1}\times\Omega_{\cI_2}}\in\mathscr{I}(\cH,\cK_1\otimes\cK_2)$ such that $\Phi^1_x=\sum_y\tr_{\cK_2}\circ\Phi_{xy}$ and $\Phi^2_y=\sum_x\tr_{\cK_1}\circ\Phi_{xy}$. Therefore, for all $X\in\cL(\cK_1)$, we have $(\Phi^1_x)^{\dagger}(X)=\sum_y\Phi^{\dagger}_{xy}(X\otimes \Id_{\cK_2})$ and for all $Y\in\cL(\cK_2)$, we have $(\Phi^2_x)^{\dagger}(Y)=\sum_x\Phi^{\dagger}_{xy}(\Id_{\cK_1}\otimes Y)$. Consider a set of operators $M=\{M(x,y,s,t):=\Phi^{\dagger}_{xy}(M_1(s)\otimes M_2(t))\}$ where $x\in\Omega_{\cI_1}$, $y\in\Omega_{\cI_2}$, $s\in\Omega_{M_1}$ and $t\in\Omega_{M_2}$. As $\Phi^{\dagger}_{xy}$ is CP for all $x,y$ and $\Phi^{\dagger}:=\sum_{x,y}\Phi^{\dagger}_{xy}$ is CP unital, we have $M(x,y,s,t)\geq 0~\forall x,y,s,t$ and $\sum_{x,y,s,t}M(x,y,s,t)=\Id_{\cH}$. Hence, $M$ is a valid measurement. Now, for all $x\in\Omega_{\cI_1}$ and $s\in\Omega_{M_1}$ we have

    \begin{align}
        \sum_{y,t}M(x,y,s,t)=& \sum_y\Phi^{\dagger}_{xy}(M_1(s)\otimes \Id_{\cK_2})\nonumber\\
        =&(\Phi^1_{x})^{\dagger}(M_1(s)).
    \end{align}
     Similarly, for all $y\in\Omega_{\cI_2}$ and $t\in\Omega_{M_2}$ we have

    \begin{align}
        \sum_{x,s}M(x,y,s,t)=& (\Phi^2_{x})^{\dagger}(M_2(t)).
    \end{align}

    Hence, $M$ is the joint measurement of $\cI^{\dagger}_1[M_1]$ and $\cI^{\dagger}_2[M_2]$. Hence, we have $\cI_1^{\dagger}[M_1]\comp\cI_2^{\dagger}[M_2]$.  
\end{proof}
The inference of this proposition is true for any number of compatible instruments. We omit that proof as this is a straight-forward generalization of Proposition \ref{Prop:ins_comp_channel_comp}.

\begin{definition}
    An incompatible set of instruments is said to be instrinsically incompatible if the set of induced measurements is compatible and  the set of induced channels is compatible.
\end{definition}

Now, consider a pair of qubit binary quantum instruments $\cJ_0=\{\Psi^0_{b_0}\}^1_{b_0=0}$ and $\cJ_1=\{\Psi^1_{b_1}\}^1_{b_1=0}$ where $\Psi^0_0=\ket{0}\bra{0}\rho\ket{0}\bra{0}$, $\Psi^0_1=\ket{1}\bra{1}\rho\ket{1}\bra{1}$, $\Psi^1_0=\frac{1}{2}\rho$, $\Psi^1_1=\frac{1}{2}\sigma_z\rho\sigma_z$ where $\{\ket{0},\ket{1}\}$ is the eigen basis of the Pauli matrix $\sigma_z$. Note that the induced measurements are $M_0=\{\ket{0}\bra{0},\ket{1}\bra{1}\}$ and $M_1=\{\frac{1}{2}\Id,\frac{1}{2}\Id\}$ respectively. Therefore, $M_0\comp M_1$ as any measurement is compatible with a trivial measurement \cite{Heinosaari_incompatibility_review_2016}. Similarly, induced channels acting on the quantum state $\rho$ can be written as $\Psi_0(\rho)=\sum^1_{x=0}\ket{i}\bra{i}\rho\ket{i}\bra{i}=\frac{1}{2}\rho+\frac{1}{2}\sigma_z\rho\sigma_z=\Psi_1(\rho)$. As measure-and-prepare channels are self-compatible \cite{Heinosaari_incomp_channel_2017}, we have $\Psi_0\comp\Psi_1$. Thus the pair $(\cJ_0,\cJ_1)$ \emph{intrinsically incompatible}. Additionally, we note the following. We know that $M_0\comp\Psi_0$ through the instrument $\cJ_0$ and $M_1\comp\Psi_1$ through the instrument $\cJ_1$. But, as $\Psi_0=\Psi_1$, we have $M_0\comp\Psi_1$ and $M_1\comp\Psi_0$. But now we show that the pair $(\cJ_0,\cJ_1)$ is incompatible. Furthermore, one can prove that the set $\{M_0,M_1,\Psi_0,\Psi_1\}$ is simultaneously implementable by a \emph{different pair of instruments}. Consider two binary quantum instruments $\cJ^{\prime}_{0}=\{\Psi^{\prime 0}_i\}$, and $\cJ^{\prime}_{1}=\{\Psi^{\prime 1}_i\}$ such that

\begin{align}
    \Psi^{\prime 0}_{b_0}=\Psi^{ 0}_{b_0}~\forall b_0\in\{0,1\};~\Psi^{\prime 1}_{b_1}=\frac{1}{2}\Psi_{1}~\forall {b_1}\in\{0,1\}
\end{align} 

Clearly, both $\cJ^{\prime}_{0}$ and $\cJ^{\prime}_{1}$ are valid instruments where $\cJ^{\prime}_{0}$ implements $M_0$, and $\Psi_0$,  and $\cJ^{\prime}_{1}$ implements $M_1$, and $\Psi_1$. Therefore, if we show that there exists a joint quantum instrument $\cJ^{\prime}=\{\Phi^{\prime}_{b_0b_1}\}$ that simultaneously implements both $\cJ^{\prime}_{0}$ and $\cJ^{\prime}_{1}$ (i.e., if the pair is compatible) then the set $\{M_0,M_1,\Psi_0,\Psi_1\}$ is simultaneously implementable. We write down the action of this joint instrument of any quantum state $\rho$ in the following.
\begin{align}
\Phi^{\prime}_{00}(\rho)&=\Phi^{\prime}_{01}(\rho)=\frac{1}{2}\bra{0}\rho\ket{0}\ket{0}\bra{0}\otimes\ket{0}\bra{0}\\
\Phi^{\prime}_{10}(\rho)&=\Phi^{\prime}_{11}(\rho)=\frac{1}{2}\bra{1}\rho\ket{1}\ket{1}\bra{1}\otimes\ket{1}\bra{1}.
\end{align}

 Note that $\cJ^{\prime}$ is \emph{not} the joint instrument of $\cJ_1$ and $\cJ_2$.  Now, although we have shown that the set $\{M_0,M_1,\Psi_0,\Psi_1\}$ is simultaneously implementable (through $\cJ^{\prime}$), interestingly the following Lemma holds.

\begin{lemma}
    The pair $(\cJ_0,\cJ_1)$ is incompatible. \label{Lem:in_incomp_in}
\end{lemma}

\begin{proof}
    Consider two qubit projective measurements $\bar{M}_0=\{\ket{0}\bra{0},\ket{1}\bra{1}\}$ and $\bar{M}_1=\{\ket{+}\bra{+},\ket{-}\bra{-}\}$ where $\{\ket{0},\ket{1}\}$ is the eigen basis of $\sigma_z$ and $\{\ket{+}.\ket{-}\}$ is the eigen basis of $\sigma_x$. Now, $\cJ^{\dagger}_0[\bar{M}_0]=\{\ket{0}\bra{0},0,0,\ket{1}\bra{1}\}$ and $\cJ^{\dagger}_1[\bar{M}_1]=\{\frac{1}{2}\ket{+}\bra{+},\frac{1}{2}\ket{-}\bra{-},\frac{1}{2}\ket{+}\bra{+},\frac{1}{2}\ket{-}\bra{-}\}$. Note that $\bar{M}_0\preceq\cJ^{\dagger}_0[\bar{M}_0]$, $\bar{M}_1\preceq\cJ^{\dagger}_1[\bar{M}_1]$. Therefore, as the pair $(\bar{M}_0,\bar{M}_1)$ is incompatible, the pair $(\cJ^{\dagger}_0[\bar{M}_0], \cJ^{\dagger}_1[\bar{M}_1])$ is incompatible. Therefore, from Proposition \ref{Prop:ins_comp_channel_comp}, we obtain that the pair $(\cJ_0,\cJ_1)$ is incompatible.
\end{proof}
 We describe the peculiarity of the pair $(\cJ_1,\cJ_2)$ in Figure \ref{fig:1}.
\begin{figure}
    \centering
    \includegraphics[width=1.0\linewidth]{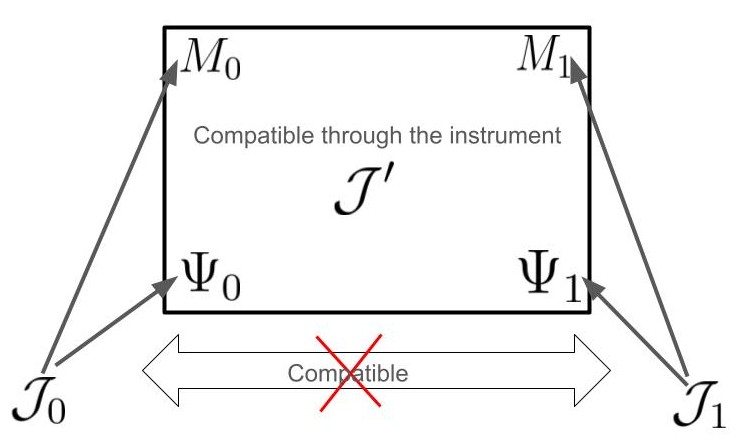}
    \caption{The instrument $\cJ_0$ induces a channel $\Psi_0$ and a measurement $M_0$. Similarly, the instrument $\cJ_1$ induces channel $\Psi_1$ and measurement $M_1$. The set $\{M_0,M_1,\Psi_0,\Psi_1\}$ is simultaneously implementable through the instrument $\cJ^{\prime}$. But $\cJ_0$ and $\cJ_1$ are \emph{incompatible}. Therefore, any certification protocol that is able to certify the incompatibility of $\cJ_0$ and $\cJ_1$, \emph{does not depend} on the incompatibility of induced measurements and channels and therefore, it is \emph{genuine}. Additionally, note that $\cJ^{\prime}$ is \emph{not} the joint instrument of $\cJ_1$ and $\cJ_2$. }
    \label{fig:1}
\end{figure}
Now, to certify a pair of incompatible instruments, we introduce a specific sequential communication task as demonstrated below.
\subsection{$\boldsymbol{(d_A,d)}$-sequential XOR task}
Consider a sequential task which is characterized by natural numbers $(n_A=2^{d},d_A,n_B=2,d_y=d,d_B=d_A,d_c=d).$  Alice is given a string of $2$ dits, denoted by $x = x_0 x_1$, chosen uniformly at random from the set of all 
  $2^{d}$ number of strings. Each component $x_y$ belongs to the alphabet $[d]$ for every $y \in [2]$. Alice sends $d_A$-dimensional message to Bob.
  After receiving the message from Alice, Bob produces an output $b_y$ given an input $y$ which is chosen uniformly at random. Then Bob sends the message of dimension $d_B$ to Charlie. After receiving the message from Bob Charlie produce $d$ possible outcomes. Here, Charlie has no inputs. The task is successful if $b_y=x_y\ominus_d(c.y)$. Mathematically, the figure of merit of this task is given by, 
 \bea\label{Eq:22xor}
\cS(d_A,d)&=&\frac{1}{2d^2}\sum_{x,y,b_y,c}\delta_{b_y,x_y\ominus_d(c.y)}p(b_y,c|x,y)\nonumber\\
&=& \frac{1}{2d^2}\sum_{x,y,c}p(b_y=x_y\ominus_d(c.y),c|x,y),
 \eea
where $\ominus_d$ denotes subtraction modulo $d$. As the task depends on the parameters $d_A$ and $d$ (all other parameters are functions of $d_A$ or $d$), we refer it as the 
$(d_A,d)$-sequential XOR task.

Now, we show that this sequential communication task can certify the incompatibility even for an intrinsically incompatible pair of instruments $(\cJ_0,\cJ_1)$ implemented by Bob. The certification of incompatibility in such an extreme case implies that this certification protocol \emph{does not depend} on the incompatibility of measurements and channels induced by the given pair of instruments and therefore, it is \emph{genuine}.
\subsection{Certification through $\boldsymbol{(2,2)}$-sequential XOR task}
\label{Subsec:2,2_XOR_task}
At first, for the convenience of the readers, we report our results for the simplest case, i.e., for $d_A=2$ and $d=2$. 

\begin{proposition}\label{prop3_1}
     In $(2,2)$ sequential XOR task, if Bob implements two compatible instruments, the optimum figure of merit $\cS^{\cQ^C}_{op}(2,2)= 3/4$. Furthermore, the optimum figure of merit using the classical strategy $\cS^\cC_{op}(2,2)=\cS^{\cQ^C}_{op}(2,2)$. 
\end{proposition}
\begin{proof}
Consider the  compatible instruments implemented by Bob to be $\cI_y=\{\Phi^y_{b_y}\}$, where $y=0,1$. The  measurement implemented by Charlie is $M := M(c)$ and the inputs of Alice is encoded to the states $\{\rho_{x_0x_1}\}$.
      The instruments $\{\cI_y\}_y$ acts on the measurement $M$ and provides the measurements $M_0:=\cI^{\dagger}_0[M]=\{M_0(b_0,c)=(\Phi^0_{b_0})^{\dagger}(M(c))\}$ and $M_1:=\cI^{\dagger}_1[M]=\{M_1(b_1,c)=(\Phi^1_{b_1})^{\dagger}(M(c))\}$. Since $\cI_0\comp\cI_1$, it follows that $M_0\comp M_1$. Therefore, there exists a joint measurement $M=\{M(b_0,c,b_1,c^{\prime})\}$ such that
$M_0(b_0,c)=\sum_{b_1,c^{\prime}}M(b_0,c,b_1,c^{\prime})$
and $M_1(b_1,c^{\prime})=\sum_{b_0,c}M(b_0,c,b_1,c^{\prime})
$. In this case, from  Eq. \eqref{Eq:22xor}, the figure of merit can be written as,
\bea\label{xor_sqc}
&& \cS^{\cQ^C}(2,2)\nonumber\\
&=&\frac18\Bigg(\rho_{00}( M_0(0,0)+M_0(0,1)+M_1(0,0)+M_1(1,1))\nonumber\\ 
&&+ \rho_{01} ( M_0(0,0)+M_0(0,1)+M_1(1,0)+M_1(0,1))\nonumber\\
&&+ \rho_{10}( M_0(1,0)+M_0(1,1)+M_1(0,0)+M_1(1,1))\nonumber\\
&&+ \rho_{11}( M_0(1,0)+M_0(1,1)+M_1(1,0)+M_1(0,1))\Bigg)\nonumber\\
&\leq& \frac18\Bigg(|| M_0(0,0)+M_0(0,1)+M_1(0,0)+M_1(1,1)||\nonumber\\ 
&&+ || M_0(0,0)+M_0(0,1)+M_1(1,0)+M_1(0,1)||\nonumber\\
&&+ || M_0(1,0)+M_0(1,1)+M_1(0,0)+M_1(1,1)||\nonumber\\
&&+ || M_0(1,0)+M_0(1,1)+M_1(1,0)+M_1(0,1)||\Bigg),
\eea
 where $||.||$ denotes the operator norm which is the highest eigenvalue for any positive semidefinite matrix. Note that for notational simplicity, we have not written the dependence of the above-said figure of merit on $\{\rho_x\}$, $\{\cI_y\}$ and $M$ as it is obvious from the context. 
Now, we look at the first term. We can write it in term of joint measurement $M$. Therefore, the first term would look like,
\bea
&&||\sum_{b_1,c'} (M(0,0,b_1,c')+M(0,1,b_1,c'))\nonumber\\
&&+\sum_{b_0,c} (M_1(b_0,c,0,0)+M_1(b_0,c,1,1))||\nonumber\\
&\leq& || M(0,0,0,0) + M(0,1,0,0) + M(0,0,1,1)\nonumber\\
&&+ M(0,1,1,1)||+||\sum_{b_1,c'} (M(0,0,b_1,c')+M(0,1,b_1,c'))\nonumber\\
&&+\sum_{(b_0,c)\neq (0,0),(0,1)} (M_1(b_0,c,0,0)+M_1(b_0,c,1,1)) ||\nonumber\\
&&\leq \tr[M(0,0,0,0) + M(0,1,0,0) + M(0,0,1,1)\nonumber\\
&&+ M(0,1,1,1)]
+ ||\sum_{b_1,c'} (M(0,0,b_1,c')+M(0,1,b_1,c'))\nonumber\\
&&+\sum_{(b_0,c)\neq (0,0),(0,1)} (M_1(b_0,c,0,0)+M_1(b_0,c,1,1)) ||\nonumber\\
&\leq& \tr[M(0,0,0,0) + M(0,1,0,0) + M(0,0,1,1)\nonumber\\
&&+ M(0,1,1,1)] + ||\Id_{\cH} ||
\eea
Clearly, similar inequalities can easily be obtained for the other three terms of \eqref{xor_sqc}. Consequently, we find,
\bea
&&\cS^{\cQ^C}(2,2)\nonumber\\
&\leq & \frac18\sum_{b_0,c,b_1,c'} \tr(M(b_0,c,b_1,c'))+ \frac18 4||\Id_{\cH}||\nonumber\\
&=& \frac18\tr(\Id_{\cH}) +\frac12\nonumber\\
&=& \frac{3}{4}.
\eea
From Proposition \ref{cor1}, we know that if a probability distribution $\{p(b_y,c|x,y)\}$ is achieved by a classical strategy, it must be achieved by a quantum strategy where the instruments implemented by Bob are compatible and therefore, $\cS^\cC_{op}(2,2)\leq\cS^{\cQ^C}_{op}(2,2)$. Therefore, to prove $\cS^\cC_{op}(2,2)=\cS^{\cQ^C}_{op}(2,2)=\frac{3}{4}$, it is enough to show that there exists a classical strategy that achieves the figure of merit $\frac{3}{4}$. 
Consider the following classical strategy: $m_A=x_0; b = m_A, \forall y; c=0, \forall m_B$
With this classical strategy, the figure of merit given in Eq. \eqref{Eq:22xor} is $\frac{3}{4}$.
\end{proof}

\begin{thm}
 The quantum advantage in the $(2,2)$-sequential XOR task implies that the pair of instruments implemented by Bob is incompatible.
\end{thm}
\begin{proof}
     From Proposition \ref{prop3_1}, we know that $\cS^\cC_{op}(2,2)=\cS^{\cQ^C}_{op}(2,2)$. Therefore, if the figure of merit using a pair instruments implemented by Bob is strictly greater than $\cS^\cC_{op}(2,2)$, it follows that the pair must be incompatible. Hence, a quantum advantage over the classical merit serves as a certification of the incompatibility of Bob's instruments. In other words, the incompatibility of instruments is necessary for the quantum advantage in $(2,2)$-XOR task.
\end{proof}

In the next proposition, we show that one can certify the incompatibility of the \emph{intrinsically incompatible} pair of instruments $(\cJ_0,\cJ_1)$ through the $(2,2)$-sequential XOR task. Therefore, it \emph{does not depend} on the incompatibility of the set of induced measurements, and the set of induced channels induced by those instruments and hence, this is a \emph{genuine} instrument incompatibility certification protocol.

\begin{proposition}\label{prop3}
 In the $(2,2)$-sequential XOR task, the optimum figure of merit $\cS^{Q}_{op}(2,2)=\frac{1}{2}(1+\frac{1}{\sqrt{2}})>\frac{3}{4}$. Furthermore, it is achieved by the intrinsically incompatible pair of qubit instruments $(\cJ_0,\cJ_1)$ implemented by Bob and therefore, the incompatibility of $(\cJ_0,\cJ_1)$ is certified.
\end{proposition}

\begin{proof}
    Let $\cI_0=\{\Phi^0_{b_0}\}$ and $\cI_1=\{\Phi^1_{b_1}\}$ be the  binary qubit instruments of Bob, and $M=\{M(c)\}$ is the  binary measurement of Charlie. 
    Given the strategy, from Eq. \eqref{Eq:22xor}, the  quantum figure of merit is 

    \begin{align}
        &\cS^\cQ(2,2)\nonumber\\
        =&\frac{1}{8}\Bigg[\tr[\rho_{00}((\Phi^0_0)^{\dagger}(M(0))+(\Phi^0_0)^{\dagger}(M(1))+(\Phi^1_0)^{\dagger}(M(0))\nonumber\\
        &+(\Phi^1_1)^{\dagger}(M(1)))]+\tr[\rho_{01}((\Phi^0_0)^{\dagger}(M(0))+(\Phi^0_0)^{\dagger}(M(1))\nonumber\\
        &+(\Phi^1_1)^{\dagger}(M(0))+(\Phi^1_0)^{\dagger}(M(1)))]+\tr[\rho_{10}((\Phi^0_1)^{\dagger}(M(0))\nonumber\\
        &+(\Phi^0_1)^{\dagger}(M(1))+(\Phi^1_0)^{\dagger}(M(0))+(\Phi^1_1)^{\dagger}(M(1)))]\nonumber\\
        &+\tr[\rho_{11}((\Phi^0_1)^{\dagger}(M(0))+(\Phi^0_1)^{\dagger}(M(1))+(\Phi^1_1)^{\dagger}(M(0))\nonumber\\
        &+(\Phi^1_0)^{\dagger}(M(1)))]\Bigg],\label{Eq:d=2_succ_prob}\\
        \leq&\frac{1}{8}[\mid\mid(\Phi^0_0)^{\dagger}(M(0))+(\Phi^0_0)^{\dagger}(M(1))+(\Phi^1_0)^{\dagger}(M(0))\nonumber\\
        &+(\Phi^1_1)^{\dagger}(M(1))\mid\mid
        +\mid\mid(\Phi^0_0)^{\dagger}(M(0))+(\Phi^0_0)^{\dagger}(M(1))\nonumber\\
        &+(\Phi^1_1)^{\dagger}(M(0))+(\Phi^1_0)^{\dagger}(M(1))\mid\mid
        +\mid\mid(\Phi^0_1)^{\dagger}(M(0))\nonumber\\
        &+(\Phi^0_1)^{\dagger}(M(1))+(\Phi^1_0)^{\dagger}(M(0))+(\Phi^1_1)^{\dagger}(M(1))\mid\mid\nonumber\\
        &+\mid\mid(\Phi^0_1)^{\dagger}(M(0))+(\Phi^0_1)^{\dagger}(M(1))+(\Phi^1_1)^{\dagger}(M(0))\nonumber\\
        &+(\Phi^1_0)^{\dagger}(M(1))\mid\mid]\label{Eq:prop3_best}.
        \end{align}
 We can simplify the above quantity and consequently, we have       
        \begin{align}\label{X}
        &\cS^\cQ(2,2)\nonumber\\
        \leq&\frac{1}{8}[\mid\mid X_0+X_1\mid\mid+\mid\mid X_0+\Id-X_1\mid\mid+\mid\mid \Id-X_0+X_1\mid\mid\nonumber\\
        &+\mid\mid \Id-X_0+\Id-X_1\mid\mid],
    \end{align}

    where $X_0:=(\Phi^0_0)^{\dagger}(M(0))+(\Phi^0_0)^{\dagger}(M(1))$ and $X_1:=(\Phi^1_0)^{\dagger}(M(0))+(\Phi^1_1)^{\dagger}(M(1))$ and $\Id$ is $2\times 2$ identity matrix. Clearly, $0\leq X_0,X_1\leq \Id$. Let, $X_i=\alpha_i\Id+\vec{c}_i.\vec{\sigma}$ where $0\leq\alpha_i\leq 1$ and $\mid\vec{c}_i\mid\leq\min\{\alpha_i,1-\alpha_i\}\leq\frac{1}{2}$. We calculate highest eigenvalues of each quantity on the right hand side of Eq. \eqref{X}.

    \begin{align}
        &\cS^\cQ(2,2)\nonumber\\
        \leq&\frac{1}{8}[(\alpha_0+\alpha_1+\mid\vec{c}_0+\vec{c}_1\mid)+(\alpha_0+1-\alpha_1+\mid\vec{c}_0-\vec{c}_1\mid)\nonumber\\
        &+(1-\alpha_0+\alpha_1+\mid\vec{c}_0-\vec{c}_1\mid)+(1-\alpha_0+1-\alpha_1+\mid\vec{c}_0+\vec{c}_1\mid)]\nonumber\\
        =&\frac{1}{4}(2+\mid\vec{c}_0+\vec{c}_1\mid+\mid\vec{c}_0-\vec{c}_1\mid)\nonumber\\
        \leq&\frac{1}{4}(2+2\sqrt{\mid c_0\mid^2+\mid c_1\mid^2})\nonumber\\
        \leq&\frac{1}{2}(1+\frac{1}{\sqrt{2}}).
    \end{align}
    Now, we show that it is achieved by the intrinsically incompatible pair of qubit instruments $(\cJ_0,\cJ_1)$ implemented by Bob. Assume that Charlie is performing the measurement $M=\{M(0)=\ket{+}\bra{+},M(1)=\ket{-}\bra{-}\}$ and Alice is encoding her inputs $\{x=x_0x_1\}$ as $\{\rho_{x_0x_1}\}$.  Then given this quantum strategy, from Eq. \eqref{Eq:d=2_succ_prob}, the corresponding figure of merit is
    \begin{align}
        &\cS^\cQ(2,2)[\cJ_0,\cJ_1,\{\rho_{x_0x_1}\},M]\nonumber\\
        =&\frac{1}{8}[\tr[\rho_{00}(\frac{1}{2}\ket{0}\bra{0}+\frac{1}{2}\ket{0}\bra{0}+\frac{1}{2}\ket{+}\bra{+}+\frac{1}{2}\ket{+}\bra{+})]\nonumber\\
        &+\tr[\rho_{01}(\frac{1}{2}\ket{0}\bra{0}+\frac{1}{2}\ket{0}\bra{0}+\frac{1}{2}\ket{-}\bra{-}+\frac{1}{2}\ket{-}\bra{-})]\nonumber\\
        &+\tr[\rho_{10}(\frac{1}{2}\ket{1}\bra{1}+\frac{1}{2}\ket{1}\bra{1}+\frac{1}{2}\ket{+}\bra{+}+\frac{1}{2}\ket{+}\bra{+})]\nonumber\\
        &+\tr[\rho_{11}(\frac{1}{2}\ket{1}\bra{1}+\frac{1}{2}\ket{1}\bra{1}+\frac{1}{2}\ket{-}\bra{-}+\frac{1}{2}\ket{-}\bra{-})]].\label{Eq:d=2_XOR_RAc_saturation}
    \end{align}
    Now, consider four pure qubit states $\ket{\psi_{00}}=\cos\frac{\pi}{8}\ket{0}+\sin\frac{\pi}{8}\ket{1}$, $\ket{\psi_{01}}=\cos\frac{\pi}{8}\ket{0}-\sin\frac{\pi}{8}\ket{1}$, $\ket{\psi_{10}}=\sin\frac{\pi}{8}\ket{0}+\cos\frac{\pi}{8}\ket{1}$, and $\ket{\psi_{11}}=\sin\frac{\pi}{8}\ket{0}-\cos\frac{\pi}{8}\ket{1}$.
    Choosing $\rho_{00}=\ket{\psi_{00}}\bra{\psi_{00}}$, $\rho_{01}=\ket{\psi_{01}}\bra{\psi_{01}}$, $\rho_{10}=\ket{\psi_{10}}\bra{\psi_{10}}$, and $\rho_{11}=\ket{\psi_{11}}\bra{\psi_{11}}$, from Eq. \eqref{Eq:d=2_XOR_RAc_saturation}, we have $\cS^\cQ(2,2)[\cJ_0,\cJ_1,\{\rho_{x_0,x_1}\},M]=\frac{1}{2}(1+\frac{1}{\sqrt{2}})=\cS^\cQ_{op}(2,2)$. This completes the proof.
\end{proof}

This certification protocol can be extended to pairs of instruments acting on $d$-dimensional Hilbert space. For that reason, we consider $(d,d)$-sequential XOR task i.e., where $d_A=d$. 

\subsection{Generalization to $\boldsymbol{(d,d)}$-sequential XOR task}
\label{Subsec:d,d_XOR_task}
At first, we state the following theorem regarding the optimum figure of merit for a given arbitrary pair of compatible instruments implemented by Bob.

\begin{thm}
   In the $(d,d)$-sequential XOR task, if Bob implements two compatible instruments, the optimum figure of merit $\cS^{\cQ^C}_{op}(d,d)= \frac{1}{2}(1+\frac{1}{d})$. Furthermore, the optimum figure of merit using the classical strategy $\cS^\cC_{op}(d,d)=\cS^{\cQ^C}_{op}(d,d)$.
    \label{Th:class_ub_d_dim}
\end{thm}

Now, consider two quantum instruments $\hat{\cJ}^0:=\{\hat{\Psi}^{0}_{b_0}:=\ket{b_0}\bra{b_0}\rho\ket{b_0}\bra{b_0}\}\in\mathscr{I}(\cH,\cH)$ and $\hat{\cJ}^1:=\{\hat{\Psi}^{1}_{b_1}:=\frac{1}{d}Z_d^{b_1}\rho (Z_d^{b_1})^{\dagger}\}\in\mathscr{I}(\cH,\cH)$ where $\dim(\cH)=d$, $\{\ket{b_0}\}^{d-1}_{b_0=0}$ is an orthonormal basis of $\cH$, and $Z_d=\sum^{d-1}_{m=0}(\omega_d)^{m}\ket{m}\bra{m}$ where $\omega_d$ is $d$th root of $1$. These instruments are implemented by Bob. Note that the measurement $\hat{M}_0=\{\ket{b_0}\bra{b_0}\}$ and $\hat{M}_1=\{\frac{1}{d}\Id_{\cH}\}$ induced by $\cI_0$ and $\cI_1$ respectively are compatible, i.e., $\hat{M}_0\comp \hat{M}_1$ as $\hat{M}_1$ is trivial. The channels $\hat{\Psi}_0=\sum_{b_0}\hat{\Psi}^0_{b_0}$ and $\hat{\Psi}_1=\sum_{b_1}\hat{\Psi}^1_{b_1}$ are compatible i.e., $\hat{\Psi}_1\comp\hat{\Psi}_2$ as $\hat{\Psi}_1=\hat{\Psi}_2$ is a measure-and-prepare channel. Now, similar as Lemma \ref{Lem:in_incomp_in}, it can be easily shown  using Proposition \ref{Prop:ins_comp_channel_comp} that the pair $(\hat{\cJ}_0,\hat{\cJ}_1)$ is incompatible. Hence, the pair $(\hat{\cJ}_0,\hat{\cJ}_1)$ is \emph{intrinsically incompatible}.

We prove in Appendix \ref{appA}, that the figure of merit $\frac{1}{2}(1+\frac{1}{\sqrt{d}})$ can be achieved by this pair of instruments. Hence, the incompatibility of $\hat{\cJ}_0$, and $\hat{\cJ}_1$ is certified.

\subsection{Certification in minimal scenario}
\label{Subsec:minimal_scenario}
In the previous section, we have presented a specific sequential communication tasks that certifies the incompatibility of any pair of $d$-dimensional quantum instruments. The simplest case was the $(2,2)$-sequential XOR task. We now investigate whether the incompatibility of instruments can be certified in an even simpler scenario where the number of Alice's inputs, i.e., $n_A<4$ and other parameters $d_A=n_B=d_y=d_B=d_c=2$ as it is not possible to consider any of these parameters less than $2$ for a non-trivial case. The choice $n_A=2$ leads to a trivial scenario, since Alice can directly encode her input in the communicated system, allowing Bob to recover it perfectly and yielding classical success probability to be $1$ and therefore, there is no quantum advantage. Therefore, the only non-trivial possibility is $n_A=3$.

Motivated by this observation, we introduce a sequential communication task with $(n_A=3,d_A=n_B=d_y=d_B=d_c=2)$, where the statistics can be written in the same manner $\{p(b_y,c|x,y)\}$ and the figure of merit is given by,
\bea\label{min}
\Bar{\cS}&=& \Big[p(0,0|0,0)+p(0,1|0,0)+p(0,0|0,1)+p(1,1|0,1)\nonumber\\
&&+p(1,0|1,0)+p(1,1|1,0)+p(0,0|1,1)+p(1,1|1,1)\nonumber\\
&&+p(0,1|2,1)+p(1,0|2,1)\Big].
\eea

\begin{proposition}\label{prop6}
    If Bob uses two compatible instruments in this task, the figure of merit $\Bar{\cS}^{\cQ^C}_{op}= 4$. Furthermore, the optimum figure of merit using the classical strategy $\Bar{\cS}^{\cC}_{op}=\Bar{\cS}^{\cQ^C}_{op}$.
\end{proposition}
For the proof of this proposition, we refer the reader to Appendix \ref{proof_6}.
\begin{thm}
 The quantum advantage in this task implies that the pair of instruments implemented by Bob is incompatible.
\end{thm}
\begin{proof}
   From the Proposition \ref{prop6}, we know that $\Bar{\cS}^\cC_{op}=\Bar{\cS}^{\cQ^C}_{op}$. Therefore, whenever $\Bar{\cS}^\cQ>\Bar{\cS}^\cC_{op}$, it follows that the instruments implemented by Bob must be incompatible.
\end{proof}
Now, in the following Proposition, we show that the incompatibility of the pair of instruments $(\cJ_0,\cJ_1)$ can be certified through this task.
\begin{proposition}\label{prop5}
   In this task, the optimum quantum figure of merit $\Bar{\cS}^\cQ_{op}= 3+\sqrt{2}>4$ and it is achieved by the intrinsically incompatible pair of instruments $(\cJ_0,\cJ_1)$ implemented by Bob.    
\end{proposition}
The proof of this proposition is postponed to Appendix \ref{proof_5}.

\section{Conclusion}
\label{sec:conc}
In this work, we have demonstrated the certification of the incompatibility of quantum instruments using sequential communication tasks. The \emph{importance} of this work lies in the following aspects. Firstly, before this work, the certification of (parallel) incompatibility of instruments has not been demonstrated to the best of our knowledge, Secondly, our certification protocol enables certification of instrument incompatibility, where the set of induced measurements and the set of induced channels are  compatible (even in minimal scenario). In other words, it \emph{does not depend} on the incompatibility of the measurements and channels induced by the instruments and therefore, it is \emph{genuine}. Thirdly, our approach is \emph{unified}; i.e., the scenarios of measurement incompatibility certification and the scenarios channel incompatibility certification become special cases of our certification scenario (see Fig. \ref{fig:communication_scenario}). Lastly, this protocol can also be viewed as an example of sequential communication task where instrument incompatibility is necessary to get quantum advantage.

Our work opens up several research directions. It is interesting to explore self-testing of quantum incompatible instruments through our sequential communication scenario. Furthermore, while the proposed certification protocol is restricted to two instruments, it can be generalized to scenarios involving more than two instruments. In such cases, different layers of incompatibility structures may arise within a set of instruments. It would be interesting to certify these structures through the construction of suitable sequential communication tasks.
Additionally, the application of these kinds of certification protocols in cryptography and secure communication task is an interesting aspect to investigate. 

\section{Acknowledgments}
This project is supported by the funding from STARS (Grant No. STARS/STARS-2/2023-0809), Government of India.

\bibliography{references}

\appendix

\section{Comparison of statistics}\label{app:compar_stat}



\begin{proposition}\label{prop:C_strict_QC}
In the scenario $(n_A,d_A,n_B,d_y,d_B,d_c)=(3,2,2,3,2,2)$,
\begin{equation}
  \Cset(3,2,2,3,2,2)\subsetneq \QCset(3,2,2,3,2,2).
\end{equation}
\end{proposition}

\begin{proof}
The inclusion relation $\Cset\subseteq\QCset$ is obtained from Proposition \ref{cor1}. To show that it is
strict, consider the qubit trine states and the ``anti-trine'' states
\begin{equation}
  |\psi_x\rangle=\cos\tfrac{\pi x}{3}|0\rangle+\sin\tfrac{\pi x}{3}|1\rangle,
  \ 
  |\psi_x^{\perp}\rangle=-\sin\tfrac{\pi x}{3}|0\rangle+\cos\tfrac{\pi x}{3}|1\rangle,
  \  x\in[3],
\end{equation}
so that $|\langle\psi_x|\psi_b^{\perp}\rangle|^2=\sin^2\frac{\pi(b-x)}{3}
=\frac{3}{4}(1-\delta_{bx})$ and
$\bar N(b):=\frac{2}{3}|\psi_b^{\perp}\rangle\langle\psi_b^{\perp}|$,
$b\in[3]$, is a POVM. Let Alice send $\rho_x=|\psi_x\rangle\langle\psi_x|$,
let Bob implement, for both $y\in[2]$, the same instrument
$\mathcal{I}_y=\{\Phi_b\}_{b\in[3]}$ with
$\Phi_b(\rho)=\mathrm{Tr}[\rho\,\bar N(b)]\,|0\rangle\langle 0|$, and let
Charlie measure $N=\{|0\rangle\langle0|,|1\rangle\langle1|\}$. The pair
$(\mathcal{I}_0,\mathcal{I}_1)$ is compatible, with joint instrument
$\Phi_{b_0b_1}(\rho)=\delta_{b_0b_1}\mathrm{Tr}[\rho\,\bar N(b_0)]\,
|0\rangle\langle0|\otimes|0\rangle\langle0|$. The resulting statistics are
\begin{equation}\label{eq:trine_stat}
  p(b,c|x,y)=\tfrac{1}{2}\,\delta_{c,0}\,(1-\delta_{bx}),
\end{equation}
which therefore belong to $\QCset(3,2,2,3,2,2)$.
Suppose these statistics belonged to $\Cset(3,2,2,3,2,2)$. Marginalising Eq.~(2) over $c$ gives, for $y=0$,
\begin{equation}
  p(b|x,0)=\sum_{m_A\in[2]} p_e(m_A|x)\,
  \Big[\sum_{m_B} p_t(b,m_B|m_A,0)\Big],
\end{equation}
i.e., the $3\times 3$ matrix $[p(b|x,0)]_{x,b}$ is the product of a $3\times 2$ and a $2\times 3$ matrix, and hence has rank at most $2$. On the other hand, by Eq.~\eqref{eq:trine_stat},
$[p(b|x,0)]_{x,b}=\frac12(J-\mathbbm{1}_{3\times 3})$, where $J$ is the all-ones matrix, and $\det\big[\frac12(J-\mathbbm{1}_{3\times 3})\big]=\frac14\neq0$, so its rank is $3$. This contradiction shows that the statistics in Eq.~\eqref{eq:trine_stat} lie in $\QCset\setminus\Cset$.
\end{proof}
The statistics of Eq.~\eqref{eq:trine_stat} do belong to $\conv\Cset$ (by Proposition~\ref{cor1}, since $d_B\geq d_c$). The strict inclusion is thus entirely due to the absence of shared randomness in Eq.~(2), and it has no consequence for optimal figures of merit.

Note that, in case (b) the compatibility of Bob's instruments is never used, so under
$n_A\leq d_A$ one has in fact $\conv\Cset=\conv\mathcal{Q}$, i.e., no
quantum advantage is possible at all. Case (a) contains the case
$d_B\geq d_c$, which covers all tasks considered in this work
(Theorems~1--3), and also contains Propositions~8 and~9 as special cases.

\begin{proposition}\label{prop:convC_strict_convQC}
In the scenario $(n_A,d_A,n_B,d_y,d_B,d_c)=(4,3,1,2,2,3)$,
\begin{equation}
  \conv\Cset(4,3,1,2,2,3)\subsetneq\conv\QCset(4,3,1,2,2,3).
\end{equation}
In particular, $\QCset(4,3,1,2,2,3)\not\subseteq\conv\Cset(4,3,1,2,2,3)$.
\end{proposition}

\begin{proof}
Since $n_B=1$, we drop the label $y$. Consider the sets of winning pairs of $(b,c)$ for each $x$. we denote these winning sets as $\{A_x\}_x$.
\begin{eqnarray}
&  A_0=\{(0,1)\},\
  A_1=\{(0,2)\},\
  A_2=\{(0,0),(1,0),(1,1)\},\nonumber \\
&  A_3=\{(0,0),(1,0),(1,2)\},
\end{eqnarray}
and the figure of merit
\begin{equation}\label{eq:S_star}
  \mathcal{S}^{\star}=\frac14\sum_{x\in[4]}\ \sum_{(b,c)\in A_x}p(b,c|x).
\end{equation}

Every $p\in\conv\Cset$ is a convex combination of deterministic strategies
$e:[4]\to[3]$, $t=(t_b,t_m):[3]\to[2]\times[2]$, $d:[2]\to[3]$. Since
$\mathcal{S}^\star$ is linear, it suffices to bound it on these. For a
deterministic strategy, the output on input $x$ is $(b,c)=f(e(x))$, where
$f(m_A):=\big(t_b(m_A),\,d(t_m(m_A))\big)$. Hence:
\begin{enumerate}
  \item[(i)] at most three distinct pairs $(b,c)$ occur, and distinct pairs
  require distinct values of $m_A\in[3]$;
  \item[(ii)] since $c=d(m_B)$ with $m_B\in[2]$, at most two distinct values
  of $c$ occur.
\end{enumerate}
These two constraints already rule out winning on every input. Suppose, for contradiction, that a deterministic strategy wins on all four inputs. Since $A_0=\{(0,1)\}$ and $A_1=\{(0,2)\}$, the outputs on $x=0,1$ are $(0,1)$ and $(0,2)$. By (ii), the values of $c$ that occur are then exactly $\{1,2\}$. The only elements of $A_2$ and $A_3$ with $c\in\{1,2\}$ are $(1,1)$ and $(1,2)$, respectively, so these must be the outputs on $x=2,3$. This requires the four distinct pairs $(0,1),(0,2),(1,1),(1,2)$, contradicting (i). Hence every deterministic strategy wins on at most three inputs, and $\mathcal{S}^\star\le\frac34$ on $\conv\Cset$. This value is achieved by the deterministic strategy $m_A=\delta_{x,0}$, $b=0$, $m_B=m_A$, $c=m_B$. It outputs $(b,c)=(0,1)$ for $x=0$ and $(0,0)$ for $x\in\{1,2,3\}$, so it wins on $x=0,2,3$.  
 Therefore
\begin{equation}
  \max_{p\in\conv\Cset}\mathcal{S}^\star=\mathcal{S}^{\star\,C}_{op}=\frac34 .
\end{equation}

 Since $n_B=1$, Bob's single instrument is
trivially compatible. There exist qutrit states $\{\rho_x\}_{x\in[4]}$, a
two-outcome instrument $\{\Phi_b\}_{b\in[2]}\subset\mathscr{I}(\mathbb{C}^3,
\mathbb{C}^2)$ and a three-outcome qubit POVM $N$ for which
\begin{equation}
  \mathcal{S}^{\star}=\frac14\sum_{x}\sum_{(b,c)\in A_x}
  \mathrm{Tr}\big[\Phi_b(\rho_x)N(c)\big]\ \geq\ 0.7668\ >\ \frac34 .
\end{equation}

Next, we show the explicit quantum  strategy.
Let $\{|0\rangle_A,|1\rangle_A,|2\rangle_A\}$ and $\{|0\rangle_B,|1\rangle_B\}$ be
the computational bases of $\mathcal{H}_A=\mathbb{C}^3$ and
$\mathcal{H}_B=\mathbb{C}^2$, respectively.

 Bob implements $\Phi_b(\rho)=K_b\,\rho\,K_b^\dagger$,
$b\in[2]$, with Kraus operators $K_b:\mathbb{C}^3\to\mathbb{C}^2$,
\begin{equation}
  K_0=|1\rangle_B\langle 0|_A+\sqrt{\tfrac{2}{3}}\,|0\rangle_B\langle 2|_A,
  \qquad
  K_1=|0\rangle_B\langle 1|_A+\sqrt{\tfrac{1}{3}}\,|1\rangle_B\langle 2|_A .
\end{equation}
Then clearly, $K^{\dagger}_{0} K_0+K_1^\dagger K_1=|0\rangle\langle0|+|1\rangle\langle1|
+(\tfrac23+\tfrac13)|2\rangle\langle2|=\Id_A$, so
$\{\Phi_b\}_{b\in[2]}\in\mathcal{I}(\mathbb{C}^3,\mathbb{C}^2)$. 
Charlie measures the qubit trine POVM
\bea
 && N(c)=\tfrac{2}{3}\,|\tau_c\rangle\langle\tau_c|, 
  |\tau_0\rangle=|0\rangle_B,\nonumber\\
  &&|\tau_{1}\rangle=\tfrac12\big(|0\rangle+\sqrt3\,|1\rangle\big)_B,
  |\tau_{2}\rangle=\tfrac12\big(|0\rangle-\sqrt3\,|1\rangle\big)_B,
\eea
which satisfies $\sum_{c}N(c)=\Id_B$.

Alice sends $\rho_x=|\psi_x\rangle\langle\psi_x|$ with
\begin{align}
|\psi_0\rangle&=\tfrac{1}{\sqrt{11}}\big(3|0\rangle+\sqrt2\,|2\rangle\big)_A, \ 
|\psi_1\rangle=\tfrac{1}{\sqrt{11}}\big(-3|0\rangle+\sqrt2\,|2\rangle\big)_A,\\
|\psi_2\rangle&=\tfrac{1}{\sqrt{26-4\sqrt{13}}}\big(3|1\rangle+(\sqrt{13}-2)|2\rangle\big)_A, \\   |\psi_3\rangle&=\tfrac{1}{\sqrt{26-4\sqrt{13}}}\big(3|1\rangle-(\sqrt{13}-2)|2\rangle\big)_A.
\end{align}

 The effective effects on Alice's system are
$G(b,c)=K_b^\dagger N(c)K_b=|g_{bc}\rangle_A\langle g_{bc}|_A$ with
\begin{equation}
\begin{aligned}
|g_{00}\rangle&=\tfrac{2}{3}|2\rangle_A, &
|g_{01}\rangle&=\tfrac{1}{\sqrt2}|0\rangle_A+\tfrac{1}{3}|2\rangle_A, \\
 |g_{02}\rangle&=-\tfrac{1}{\sqrt2}|0\rangle_A+\tfrac{1}{3}|2\rangle_A,&
|g_{10}\rangle&=\sqrt{\tfrac23}\,|1\rangle_A, \nonumber\\
|g_{11}\rangle&=\tfrac{1}{\sqrt6}|1\rangle_A+\tfrac{1}{\sqrt6}|2\rangle_A, &
|g_{12}\rangle&=\tfrac{1}{\sqrt6}|1\rangle_A-\tfrac{1}{\sqrt6}|2\rangle_A ,
\end{aligned}
\end{equation}
and $p(b,c|x)=|\langle g_{bc}|\psi_x\rangle|^2$. Explicitly:
\begin{center}
\renewcommand{\arraystretch}{1.4}
\begin{tabular}{c|ccc|ccc}
 & \multicolumn{3}{c|}{$b=0$} & \multicolumn{3}{c}{$b=1$}\\
$x$ & $c=0$ & $c=1$ & $c=2$ & $c=0$ & $c=1$ & $c=2$\\ \hline
$0$ & $\frac{8}{99}$ & $\mathbf{\frac{11}{18}}$ & $\frac{49}{198}$ & $0$ & $\frac{1}{33}$ & $\frac{1}{33}$\\
$1$ & $\frac{8}{99}$ & $\frac{49}{198}$ & $\mathbf{\frac{11}{18}}$ & $0$ & $\frac{1}{33}$ & $\frac{1}{33}$\\
$2$ & $\mathbf{\frac{2}{9}-\frac{4\sqrt{13}}{117}}$ & $\frac{1}{18}-\frac{\sqrt{13}}{117}$ & $\frac{1}{18}-\frac{\sqrt{13}}{117}$
    & $\mathbf{\frac13+\frac{2\sqrt{13}}{39}}$ & $\mathbf{\frac16+\frac{\sqrt{13}}{26}}$ & $\frac16-\frac{\sqrt{13}}{26}$\\
$3$ & $\mathbf{\frac{2}{9}-\frac{4\sqrt{13}}{117}}$ & $\frac{1}{18}-\frac{\sqrt{13}}{117}$ & $\frac{1}{18}-\frac{\sqrt{13}}{117}$
    & $\mathbf{\frac13+\frac{2\sqrt{13}}{39}}$ & $\frac16-\frac{\sqrt{13}}{26}$ & $\mathbf{\frac16+\frac{\sqrt{13}}{26}}$
\end{tabular}
\end{center}
Each row sums to $1$, and the entries in bold are those with
 $(b,c)\in A_x$. 
 The winning probabilities are
\begin{align}
  \Pr[\mathrm{win}|x]&=\frac{11}{18}\ \ \forall x=\{0,1\}\\
  \Pr[\mathrm{win}|x]&=\frac{13+\sqrt{13}}{18}\ \ \forall x=\{2,3\},
\end{align}
Therefore, the success merit reads as
\begin{equation}
  \mathcal{S}^\star=\frac14\Big(2\cdot\frac{11}{18}+2\cdot\frac{13+\sqrt{13}}{18}\Big)
  =\frac{24+\sqrt{13}}{36}\approx0.76682>\frac34 .
\end{equation}
Therefore the corresponding statistics lie in $\QCset\setminus\conv\Cset$,
and $\conv\Cset\subsetneq\conv\QCset$.
\end{proof}

\begin{remark}
    The scenario $(4,3,1,2,2,3)$ violates both conditions of
Proposition~\ref{cor1} minimally, since
$d_B=2<3=\min\{d_c,n_Bd_A\}$ and $n_A=4>3=d_A$. The same conclusion holds with $n_B=2$: adding to Bob a second, one-outcome instrument $\rho\mapsto\mathrm{Tr}(\rho)\,\sigma_0$, which is compatible with any instrument, and assigning zero weight to $y=1$ leaves both classical and quantum merit
unchanged. Consequently, when $d_B<d_c$, a quantum advantage in a sequential communication task does \emph{not}, in general, certify the incompatibility of Bob's instruments; this is why the certification protocols of this work require $d_B\geq d_c$.
\end{remark}


\section{Proof of Theorem \ref{Th:class_ub_d_dim} and the certification of intrinsically incompatible instruments in $(d,d)$-sequential XOR task}\label{appA}

\subsection{Proof of Theorem \ref{Th:class_ub_d_dim}}
Consider the  compatible instruments implemented by Bob to be $\cI_y=\{\Phi^y_{b_y}\}$, where $y=0,1$. The  measurement for Charlie is $M := \{M(c)\}$.
After action of instruments on measurement $M$,  it provides the measurements $M_0:=\cI^{\dagger}_0[M]=\{M_0(b_0,c)=(\Phi^0_{b_0})^{\dagger}(M(c))\}$ and $M_1:=\cI^{\dagger}_1[M]=\{M_1(b_1,c)=(\Phi^1_{b_1})^{\dagger}(M(c))\}$. As $\cI_1\comp\cI_2$, we have $M_1\comp M_2$. Hence, there exists a joint measurement $M=\{M(b_0,c,b_1,c^{\prime})\}$ such that $M_0(b_0,c)=\sum_{b_1,c^{\prime}}M(b_0,c,b_1,c^{\prime})$ and $M_1(b_1,c^{\prime})=\sum_{b_0,c}M(b_0,c,b_1,c^{\prime})$.
 If Alice uses the states $\{\rho_{x_0 x_1}\}$ as encoding then the  figure of merit can be written as
 \begin{align}
      &\cS^{\cQ^C}(d,d)\nonumber\\
    =& \frac{1}{2d^2}\sum_{x}\big[\sum_c(\tr[\rho_{x_0x_1}(\Phi^0_{x_0})^{\dagger}(M(c))+(\Phi^{1}_{x_1\ominus_d c})^{\dagger}(M(c)))]\big]\nonumber\\
    =&\frac{1}{2d^2}\sum_{x}\big[\tr[\rho_{x_0x_1}\sum_c(M_0(x_0,c)+M_1(x_1\ominus_d c,c))]\big]\nonumber\\
    \leq&\frac{1}{2d^2}\sum_{x}\big[||\sum_c(M_0(x_0,c)+M_1(x_1\ominus_d c,c))||\big]\nonumber\\
    =&\frac{1}{2d^2}\sum_{x}\big[||\sum_c(\sum_{b_1,\tilde{c}^{\prime}}M(x_0,c,b_1,\tilde{c}^{\prime})+\sum_{b_0,\tilde{c}}M(b_0,\tilde{c},x_1\ominus_d c,c))||\big]\nonumber\\
    =&\frac{1}{2d^2}\sum_{x}\big[||\sum_c(\sum_{b_1,\tilde{c}^{\prime}}M(x_0,c,b_1,\tilde{c}^{\prime})+\sum_{b_0,\tilde{c}}M(b_0,\tilde{c},x_1\ominus_d c,c))||\big]\nonumber\\
    =&\frac{1}{2d^2}\sum_{x}\big[||\sum_c(M(x_0,c,x_1\ominus_d c,c)\nonumber\\
&+\sum_{b_1,\tilde{c}^{\prime}}M(x_0,c,b_1,\tilde{c}^{\prime})+\sum_{\substack{b_0,\tilde{c}\\
    b_0\neq x_0,\tilde{c}\neq c}}M(b_0,\tilde{c},x_1\ominus_d c,c)||\big]\nonumber\\
    \leq&\frac{1}{2d^2}\sum_{x}\big[||\sum_cM(x_0,c,x_1\ominus_d c,c)||\nonumber\\
&+||\sum_c(\sum_{b_1,\tilde{c}^{\prime}}M(x_0,c,b_1,\tilde{c}^{\prime})+\sum_{\substack{b_0,\tilde{c}\\
    b_0\neq x_0,\tilde{c}\neq c}}M(b_0,\tilde{c},x_1\ominus_d c,c))||\big]\nonumber
    \end{align}
    \begin{align}
    \leq&\frac{1}{2d^2}\sum_{x_0,x_1}\big[\tr[\sum_cM(x_0,c,x_1\ominus_d c,c)]\nonumber\\      &+||\sum_{b_0,c,b_1,\tilde{c}^{\prime}}M(b_0,c,b_1,\tilde{c}^{\prime})||\big]\nonumber\\
    \leq&\frac{1}{2d^2}\big[\tr[\sum_{x_0,x_1,c}M(x_0,c,x_1\ominus_d c,c)]+\sum_{x_0,x_1}||\Id_{\cH}||\big]\nonumber\\ 
    \leq&\frac{1}{2d^2}\big[\tr[\Id_{\cH}]+\sum_{x_0,x_1}||\Id_{\cH}||\big]\nonumber\\
    =&\frac{1}{2d^2}\big[d+d^2\big]
    =\frac{1}{2}\big[1+\frac{1}{d}\big].
 \end{align}

 From Proposition \ref{cor1}, we know that if a probability distribution $\{p(b_y,c|x,y)\}$ is achieved by a classical strategy, it must be achieved by a quantum strategy where the instruments implemented by Bob are compatible and therefore, $\cS^\cC_{op}(d,d)\leq\cS^{\cQ^C}_{op}(d,d)$. Therefore, to prove $\cS^\cC_{op}(d,d)=\cS^{\cQ^C}_{op}(d,d)=\frac{1}{2}(1+\frac{1}{d})$, it is enough to show that there exists a classical strategy that achieves the figure of merit $\frac{1}{2}(1+\frac{1}{d})$.
 Now, consider a classical strategy where, given input $x$, Alice always sends $x_0$ to Bob, Charlie always announces his output to be zero, and for input $y=0$ Bob announces the output $x_0$ and for input $y=1$ he randomly announces the output from the set $\{0,\ldots,d-1\}$ with uniform probability $\frac{1}{d}$. In other words, $p_e(m_A=x_0|x)=1\forall x$, and for all $r\in\{0,\ldots,d-1\}$ $p_d(c=0|m_B)=1~\forall m_B$, $p_t(b_0=x_0,0|x,0)=1~\forall x$, and $p_t(b_1=r,0|x,1)=\frac{1}{d}~\forall x$. Clearly, since each input of Bob occurs with probability $\frac{1}{2}$, the figure of merit is $\frac{1}{2}(1+\frac{1}{d})$.

\subsection{The certification of instrinsically incompatible instruments in $(d,d)$-sequential XOR task}

Recall that we are given two quantum instruments $\hat{\cJ}^0:=\{\hat{\Psi}^{0}_{b_0}:=\ket{b_0}\bra{b_0}\rho\ket{b_0}\bra{b_0}\}\in\mathscr{I}(\cH,\cH)$ and $\hat{\cJ}^1:=\{\hat{\Psi}^{1}_{b_1}:=\frac{1}{d}Z_d^{b_1}\rho (Z_d^{b_1})^{\dagger}\}\in\mathscr{I}(\cH,\cH)$ where $\dim(\cH)=d$, $\{\ket{b_0}\}^{d-1}_{b_0=0}$ is an orthonormal basis of $\cH$, and $Z_d=\sum^{d-1}_{m=0}(\omega_d)^{m}\ket{m}\bra{m}$ where $\omega_d$ is $d$th root of $1$. Now, consider a projective measurement $M=\{M(c)=\ket{f^d_c}\bra{f^d_c}\}$ where $\{\ket{f^d_c}=\frac{1}{\sqrt{d}}\sum^{d-1}_{m=0}\omega^{mc}_d\ket{m}\}^{d-1}_{c=0}$ is an orthonormal basis of $\cH$. This measurement is implemented by Charlie. Clearly, the pair of bases $\{\ket{b_0}\}^{d-1}_{b_0=0}$ and $\{\ket{f^d_c}\}^{d-1}_{c=0}$ are mutually unbiased basis (MuB). It can be easily seen that $Z^{b_1}\ket{f^d_c}=\ket{f^d_{b_1\oplus_d c}}$ where $\oplus_d$ denotes addition modulo $d$. Therefore, $(\hat{\Psi}^0_{b_0})^{\dagger}(M(c))=\frac{1}{d}\ket{b_0}\bra{b_0}$ and $(\hat{\Psi}^1_{b_1})^{\dagger}(M(c))=\frac{1}{d}\ket{f^d_{b_1\oplus_d c}}\bra{f^d_{b_1\oplus_d c}}$ for all $b_0,b_1,c$.

Now, consider the figure of merit of this task described in Eq. \eqref{Eq:22xor}. Given Bob and Charlie are performing the above-said instruments and measurement, respectively, and Alice is encoding the string $x$ to $\{\rho_{x_0x_1}\}$, where $\rho_{x_0x_1}=\ket{\psi_{x_0x_1}}\bra{\psi_{x_0x_1}}$ where $\ket{\psi_{x_0x_1}}=\frac{\braket{f^d_{x_1}\mid x_0}}{|\braket{f^d_{x_1}\mid x_0}|}\ket{f^d_{x_1}}+\ket{x_0}$ for all $x_0,x_1$. The figure of merit in this strategy takes the form as follows. 

\begin{align}
    &\cS^\cQ(d,d)\{\hat{\cJ}^0, \hat{\cJ}^1,\{\rho_{x_0x_1}\}\}\nonumber\\
    =& \frac{1}{2d^2}\sum_{x}\big[\sum_c\tr[\rho_{x_0x_1}(\hat{\Psi}^0_{x_0})^{\dagger}(M(c))]\nonumber\\
    &\hspace{3.4cm}+\sum_c\tr[\rho_{x_0x_1}(\hat{\Psi}^{1}_{x_1\ominus_d c})^{\dagger}(M(c))]\big]\nonumber\\
    =& \frac{1}{2d^2}\sum_{x}\big[\sum_c\tr[\rho_{x_0x_1}\frac{1}{d}\ket{x_0}\bra{x_0}]+\sum_c\tr[\rho_{x_0x_1}\frac{1}{d}\ket{f^d_{x_1}}\bra{f^d_{x_1}}]\big]\nonumber\\
    =& \frac{1}{2d^2}\sum_{x}\big[\tr[\rho_{x_0x_1}\ket{x_0}\bra{x_0}+\ket{f^d_{x_1}}\bra{f^d_{x_1}}]\big].
    \end{align}
    Putting $\rho_{x_0x_1}=\ket{\psi_{x_0x_1}}\bra{\psi_{x_0x_1}}$ (defined above) into the above equation, we get
    \begin{align}
    \cS^\cQ(d,d)[\hat{\cJ}^0, \hat{\cJ}^1,\{\rho_{x_0x_1}\},M]=& \frac{1}{2d^2}\sum_{x_0,x_1}\big[1+|\braket{f^d_{x_1}\mid x_0}|\big]\nonumber\\
    =& \frac{1}{2d^2}\sum_{x_0,x_1}\big[1+\frac{1}{\sqrt{d}}\big]\nonumber\\
    =& \frac{1}{2}\big[1+\frac{1}{\sqrt{d}}\big].
\end{align}
 As $\cS^\cQ(d,d)[\hat{\cJ}^0, \hat{\cJ}^1,\{\rho_{x_0x_1}\},M]>\cS^{\cQ^C}_{op}(d,d)=\cS^\cC_{op}(d,d)$, the quantum advantage is non zero which implies the incompatibility of any pair $\hat{\cJ}_0$ and $\hat{\cJ}_1$ used by Bob. Hence, the incompatibility of $\hat{\cJ}_0$ and $\hat{\cJ}_1$ is certified.
\section{Certification in minimal scenario}
\subsection{Proof of Proposition \ref{prop6}}\label{proof_6}
Consider the  compatible instruments implemented by Bob to be $\cI_y=\{\Phi^y_{b_y}\}$, where $y=0,1$. The  measurement for Charlie is $M := M(c)$.
      The instruments $\{\cI_y\}_y$ acts on the measurement $M$ and provide the measurements $M_0:=\cI^{\dagger}_0[M]=\{M_0(b_0,c)=(\Phi^0_{b_0})^{\dagger}(M(c))\}$ and $M_1:=\cI^{\dagger}_1[M]=\{M_1(b_1,c)=(\Phi^1_{b_1})^{\dagger}(M(c))\}$. Since $\cI_0\comp\cI_1$, it follows that $M_0\comp M_1$. Therefore, there exists a joint measurement $M=\{M(b_0,c,b_1,c^{\prime})\}$ such that
$M_0(b_0,c)=\sum_{b_1,c^{\prime}}M(b_0,c,b_1,c^{\prime})$
and
$
M_1(b_1,c^{\prime})=\sum_{b_0,c}M(b_0,c,b_1,c^{\prime}).
$. Now, if Alice encodes her inputs in the states $\{\rho_{x_0x_1}\}$ then we can write the  merit in terms of $M_0(b_0,c)$ and $M_1(b_1,c)$ as 
\bea
&&\Bar{\cS}^{\cQ^C}\nonumber\\
&=&\tr\Big[\rho_0\Big((\Phi^0_0)^\dagger M(0) +(\Phi^0_0)^\dagger M(1)+(\Phi^1_0)^\dagger M(0) \nonumber\\
&&+(\Phi^1_1)^\dagger M(1)\Big)
+ \rho_1\Big((\Phi^0_1)^\dagger M(0) +(\Phi^0_1)^\dagger M(1)\nonumber\\
&&+(\Phi^1_0)^\dagger M(0) +(\Phi^1_1)^\dagger M(1)\Big)\nonumber\\
&&+ \rho_2\left((\Phi^1_0)^\dagger M(1) +(\Phi^1_1)^\dagger M(0)\right)\Big]\nonumber
\eea
\bea
&\leq& ||M_0(0,0) + M_0(0,1)+M_1(0,0) + M_1(1,1)||\nonumber\\
&&+ || M_0(1,0) + M_0(1,1)+M_1(0,0) + M_1(1,1)||\nonumber\\
&&+ || M_1(0,1) + M_1(1,0)||\nonumber\\
&=&||\sum_{b_0,c}(M(b_0,c,0,0) + M(b_0,c,1,1))\nonumber\\
&&+ \sum_{b_1,c'}(M(0,0,b_1,c') + M(0,1,b_1,c'))||\nonumber\\
&&+||\sum_{b_0,c}(M(b_0,c,0,0) + M(b_0,c,1,1))\nonumber\\
&&+ \sum_{b_1,c'}(M(1,0,b_1,c') + M(1,1,b_1,c'))||\nonumber\\
&&+||\sum_{b_0,c}(M(b_0,c,0,1) + M(b_0,c,1,0))||\nonumber\\
&=&||M(0,0,0,0)+M(0,1,0,0)+M(0,0,1,1)\nonumber\\
&&+M(0,1,1,1)
+\sum_{b_0,c}(M(b_0,c,0,0) + M(b_0,c,1,1))\nonumber\\
&&+ \sum_{(b_1,c')\neq (0,0),(1,1)}(M(0,0,b_1,c') + M(0,1,b_1,c'))||\nonumber\\
&&+||M(1,0,0,0)+M(1,0,1,1)+M(1,1,0,0)\nonumber\\
&&+M(1,1,1,1)
+\sum_{b_0,c}(M(b_0,c,0,0) + M(b_0,c,1,1))\nonumber\\
&&+ \sum_{(b_1,c')\neq (0,0),(1,1)}(M(1,0,b_1,c') + M(1,1,b_1,c'))||\nonumber\\
&&+||\sum_{b_0,c}(M(b_0,c,0,1) + M(b_0,c,1,0))||\nonumber\\
&\leq &\tr\Big[M(0,0,0,0)+M(0,1,0,0)+M(0,0,1,1)\nonumber\\
&&+M(0,1,1,1)
+M(1,0,0,0)+M(1,0,1,1)\nonumber\\
&&+M(1,1,0,0)+M(1,1,1,1)\nonumber\\
&&+\sum_{b_0,c}(M(b_0,c,0,1) + M(b_0,c,1,0))\Big]\nonumber\\
&&+||\sum_{b_0,c}(M(b_0,c,0,0) + M(b_0,c,1,1))\nonumber\\
&&+\sum_{(b_1,c')\neq (0,0),(1,1)}(M(1,0,b_1,c') + M(1,1,b_1,c'))||\nonumber\\
&&+||\sum_{b_0,c}(M(b_0,c,0,0) + M(b_0,c,1,1))\nonumber\\
&&+ \sum_{(b_1,c')\neq (0,0),(1,1)}(M(1,0,b_1,c') + M(1,1,b_1,c'))||\nonumber\\
&\leq& \tr[\Id_{\cH}]+2||\Id_{\cH}||\nonumber\\
&=& 4.
\eea
From Proposition \ref{cor1}, we know that if a probability distribution $\{p(b_y,c|x,y)\}$ is achieved by a classical strategy, it must be achieved by a quantum strategy where the instruments implemented by Bob are compatible and therefore, $\bar{\cS}^\cC_{op}\leq\bar{\cS}^{\cQ^C}_{op}$. Therefore, to prove $\bar{\cS}^\cC_{op}=\bar{\cS}^{\cQ^C}_{op}=4$, it is enough to show that there exists a classical strategy that achieves the figure of merit $4$.
 Now, we show that there exists a classical strategy which achieves this figure of merit.
Consider a classical strategy where  $p_e(m_A=0|x=0)=p_e(m_A=0|x=1)=p_e(m_A=1|x=2)=1$, $p_t(b_0=0,m_B=0|m_A=0,y=0) = p_t(b_1=0,m_B=0|m_A=0,y=1)=p_t(b_1=1,m_B=0|m_A=1,y=1)=1$, and $p_d(c=0|m_B=0)=1$. With this strategy, from \eqref{min} and Eq.\eqref{Eq:prob_class_strat}, the figure of merit is $4$.

\subsection{Proof of Proposition \ref{prop5}}\label{proof_5}
Given an input $x$, let $\rho_x$ be the  state prepared by Alice, $\cI_0=\{\Phi^0_x\}$ and $\cI_1=\{\Phi^1_x\}$ be binary qubit instruments implemented by Bob, and $M=\{M(c)\}$ be the dichotomic measurement implemented by Charlie. The  quantum figure of merit corresponding to this strategy can be written from Eq. \eqref{min} as follows:
\bea
&&\Bar{\cS}^\cQ\nonumber\\
&=& \tr\Big[\rho_0\Big((\Phi^0_0)^\dagger M(0) +(\Phi^0_0)^\dagger M(1)+(\Phi^1_0)^\dagger M(0) \nonumber\\
&&+(\Phi^1_1)^\dagger M(1)\Big)+ \rho_1\Big((\Phi^0_1)^\dagger M(0) +(\Phi^0_1)^\dagger M(1)\nonumber\\
&&+(\Phi^1_0)^\dagger M(0) +(\Phi^1_1)^\dagger M(1)\Big)\nonumber\\
&&+ \rho_2\left((\Phi^1_0)^\dagger M(1) +(\Phi^1_1)^\dagger M(0)\right)\Big]\label{Eq:Q_suc_mer_min}\\
&\leq& ||(\Phi^0_0)^\dagger M(0) +(\Phi^0_0)^\dagger M(1)+(\Phi^1_0)^\dagger M(0) +(\Phi^1_1)^\dagger M(1)||\nonumber\\
&&+ ||(\Phi^0_1)^\dagger M(0) +(\Phi^0_1)^\dagger M(1)+(\Phi^1_0)^\dagger M(0) +(\Phi^1_1)^\dagger M(1)||\nonumber\\
&&+ ||(\Phi^1_0)^\dagger M(1) +(\Phi^1_1)^\dagger M(0)||\nonumber\\
&=& ||X_0+X_1||+||X_1+\Id-X_0||+||\Id-X_1||\label{sq1},
\eea
where $X_0=(\Phi^0_0)^\dagger M(0) +(\Phi^0_0)^\dagger M(1)$ and $ X_1=(\Phi^1_0)^\dagger M(0) +(\Phi^1_1)^\dagger M(1)$. Let, $X_i=\alpha_i\Id+\vec{c}_i.\vec{\sigma}$ where $0\leq\alpha_i\leq 1$ and $\mid\vec{c}_i\mid\leq\min\{\alpha_i,1-\alpha_i\}\leq\frac{1}{2}$. This indicates $\sum_i |c_i|^2\leq1/2$. Now, we can write the above quantity as,
\bea\label{sqq}
&&\Bar{\cS}^\cQ\nonumber\\
&\leq & (\alpha_0+\alpha_1+|\vec{c}_0+\vec{c}_1|+\alpha_1+1-\alpha_0\nonumber\\
&&+|\vec{c}_1-\vec{c}_0|+1-\alpha_1+|\vec{c}_1|)\nonumber\\
&=& (2 +\alpha_1+|\vec{c}_1|+|\vec{c}_0+\vec{c}_1|+|\vec{c}_1-\vec{c}_0|)\nonumber\\
&\leq& (2 +\alpha_1+|\vec{c_1}|+2\sqrt{|c_0|^2+|c_1|^2})\nonumber\\
&\leq& 2+\sqrt{2}+\alpha_1+|\vec{c}_1|\nonumber\\
&\leq& 3+\sqrt{2}.
\eea
Last line comes from the fact that $|\vec{c}_1|\leq\min\{\alpha_1,1-\alpha_1\},$ which implies $\alpha_1+|\vec{c}_1|\leq 1$.

 Now, we show that the upper bound is saturated for the intrinsically incompatible pair of qubit instruments $(\cJ_0,\cJ_1)$ implemented by Bob. Assume that Charlie is performing the measurement $M=\{M(0)=\ket{+}\bra{+},M(1)=\ket{-}\bra{-}\}$.  Then from \eqref{Eq:Q_suc_mer_min}, we have,
 \bea
&&\Bar{\cS}^\cQ[\cJ_0,\cJ_1,\{\rho_x\},M]\nonumber\\
&=& \tr\Bigg[\rho_0\left(\frac12\ket{+}\bra{+} +\frac12\ket{+}\bra{+} +\frac12\ket{0}\bra{0} +\frac12\ket{0}\bra{0}\right)\nonumber\\
&+& \rho_1\left(\frac12\ket{+}\bra{+} +\frac12\ket{+}\bra{+} +\frac12\ket{1}\bra{1} +\frac12\ket{1}\bra{1}\right)\nonumber\\
&+& \rho_2\left(\frac12\ket{-}\bra{-} +\frac12\ket{-}\bra{-}\right)\Bigg].
 \eea
 We chose $\rho_0=\ket{\psi_0}\bra{\psi_0},\rho_1=\ket{\psi_1}\bra{\psi_1}$ and $\rho_2=\ket{\psi_2}\bra{\psi_2}$, where $\ket{\psi_0}=\cos\frac{\pi}{8}\ket{0}+\sin\frac{\pi}{8}\ket{1},\ket{\psi_1}=\sin\frac{\pi}{8}\ket{0}+\cos\frac{\pi}{8}\ket{1}$ and $\ket{\psi_2}=\ket{-}$. With this strategy, $\Bar{\cS}^\cQ[\cJ_1,\cJ_2,\{\rho_x\},M]=3+\sqrt{2}$. This completes the proof.
 \section{Necessity of 'no input on Charlie' scenario}
 \label{App_sec:Charlie}
In this section, we prove that if Charlie is allowed inputs with unrestricted measurements, quantum advantage in a sequential communication task does not, in general, certify Bob's instrument incompatibility. Specifically, we show that, when Charlie has access to inputs, there exists a sequential communication task that exhibits a quantum advantage even when Bob employs two compatible quantum instruments. In this case, the probability distribution takes the form $\{p(b_y,c|x,y,z)\}$, where $z$ denotes the input of Charlie and $z\in[n_C]$. Now consider a communication task, where $n_A=n_B=n_C=d_c=d_A=d_B=2$ and $d_y=1$. As $d_y$ is fixed, we can omit this from the probability distribution. Consider the following figure of merit written in linear combination of terms $p(c|x,y,z)$,
 \bea\label{merit_s'}
\cS' &=& p(0|0,0,0)+p(0|0,0,1)+p(0|1,0,0)\nonumber\\
&&+p(1|1,0,1)+p(1|0,1,0).
 \eea
 Now we calculate the optimal classical figure of merit for this task.
Using the Eq. \eqref{c_stat}, from \eqref{merit_s'}, we can write,
\bea
\cS'^\cC_{op}&=&\max_{\substack{\{p_e(m_A|x)\},\\ \{p_t(m_B|y,m_A)\},\\ \{p_d(c|m_B)\} }} \sum_{m_A,m_B} \Big(p_e(m_A|0)p_t(m_B|0,m_A)p_d(0|0,m_B)\nonumber\\
&&+p_e(m_A|0)p_t(m_B|0,m_A)p_d(0|1,m_B)\nonumber\\
&&+p_e(m_A|1)p_t(m_B|0,m_A)p_d(0|0,m_B)\nonumber\\
&&+p_e(m_A|1)p_t(m_B|0,m_A)p_d(1|1,m_B)\nonumber\\
&&+p_e(m_A|0)p_t(m_B|1,m_A)p_d(1|0,m_B)\Big)\nonumber
\eea
\bea
&=&\max_{\substack{\{p_e(m_A|x)\},\\ \{p_t(m_B|y,m_A)\} }} \sum_{m_B}\max\Big\{\sum_{m_A}p_e(m_A|0)p_t(m_B|0,m_A)\nonumber\\
&&+\sum_{m_A}p_e(m_A|1)p_t(m_B|0,m_A), \nonumber\\
&&\sum_{m_A}p_e(m_A|0)p_t(m_B|1,m_A)\Big\}\nonumber\\
&&+\sum_{m_B}\max\Big\{\sum_{m_A}p_e(m_A|0)p_t(m_B|0,m_A), \nonumber\\
&&\sum_{m_A}p_e(m_A|1)p_t(m_B|0,m_A)\Big\}
\eea

We check all the deterministic strategy to find the maximum value of above equation. Consequently, we find
$\cS'^\cC_{op}\leq 4$. One optimum strategy is $m_A=0, \forall x; m_B=y,\forall m_A$ and $c=m_B,\forall z$.

\begin{proposition}
    If Charlie's measurements are compatible and Bob uses compatible instruments as before, the value of the merit of \eqref{merit_s'},$\Bar{\cS}^{\cQ^C}\leq 4$.
\end{proposition}
\begin{proof}
    Let us take two compatible measurements $M(l)$ and $N(k)$ which are implemented by Charlie. As they are compatible, there exists a measurement $X(l,k)$ such that $\sum_k X(l,k)=M(l)$ and $\sum_l X(l,k)=N(k)$. The compatible instruments of Bob are $\cI_y=\{\Phi^y\}$. The instruments $\{\cI_y\}_y$ acts on the measurement $X$ and provide the measurements $X_0 :=\cI_0^\dagger [X]= \{X_0(l,k)= (\Phi^0)^\dagger X(l,k)\}$ and $X_1 :=\cI_1^\dagger [X]= \{X_1(l,k)= (\Phi^1)^\dagger X(l,k)\}$. Since $\cI_0\comp\cI_1$, it follows that $X_0\comp X_1$. Therefore, there exists a joint measurement $X=\{X(l,k,l',k')\}$ such that
$X_0(l,k)=\sum_{l',k'}X(l,k,l',k')$
and
$
X_1(l',k')=\sum_{l,k}X(l,k,l',k')$. Now, if Alice encodes her inputs in the states $\{\rho_{x}\}$ then we can write the  merit as ,
    \bea
&& \cS'^{\cQ^C}\nonumber\\
&=&\tr [\rho_0 ((\Phi^0)^\dagger M(0)+(\Phi^0)^\dagger N(0)+(\Phi^1)^\dagger M(1))\nonumber\\
&&+ \rho_1 ((\Phi^0)^\dagger M(0)+(\Phi^0)^\dagger N(1))]\nonumber\\
&\leq& || \sum_k X_0(0,k) + \sum_l X_0(l,0) + \sum_k X_1(1,k) ||\nonumber\\
&&+ ||\sum_k X_0(0,k) +\sum_l X_0 (l,1)||\nonumber\\
&=& || \sum_{k,l',k'} X(0,k,l',k') + \sum_{l,l',k'} X(l,0,l',k') + \sum_{l,k,k'} X(l,k,1,k') ||\nonumber\\
&&+ ||\sum_{k,l',k'} X(0,k,l',k') +\sum_{l,l',k'} X (l,1,l',k')||\nonumber\\
&=& ||\sum_{k'} X(0,1,0,k') + \sum_{l,l',k'} X(l,0,l',k') + \sum_{l,k,k'} X(l,k,1,k')||\nonumber\\
&&+||\sum_{k'} X(0,0,1,k')+ \sum_{k'} X(0,0,0,k')+\sum_{k'} X(0,1,1,k') ||\nonumber\\
&&+ ||\sum_{l',k'} X(0,0,l',k')+\sum_{l,l',k'} X (l,1,l',k')||\nonumber\\
&&+ ||\sum_{l',k'} X(0,1,l',k')||\nonumber\\
&\leq& 4||\Id_{\cH}||=4.
    \eea
\end{proof}
Consider a pair of qubit one-outcome quantum instruments $\cJ'_0=\Psi'^0_0=\rho$ and $\cJ'_1=\Psi'^1_0=\tr(\rho)\ket{-}\bra{-}$ .

Alice encodes her inputs into two states $\rho_0$ and $\rho_1$. Charlie executes measurement $\sigma_x$ for $z=0$ and $\sigma_z$ for $z=1$. Therefore, quantum figure of merit from \eqref{merit_s'},
\bea
&&\cS'^\cQ[\cJ'_0,\cJ'_1,\{\rho_0,\rho_1\},\{\sigma_x,\sigma_z\}]\nonumber\\
&=& \tr\Big[\rho_0 (\ket{+}\bra{+}+\ket{0}\bra{0})
+\rho_1(\ket{+}\bra{+}+\ket{1}\bra{1})\nonumber\\
&&+\ket{-}\bra{-}\ket{-}\bra{-}\Big]
\eea
 Using $\rho_0=\ket{\psi_0}\bra{\psi_0},\rho_1=\ket{\psi_1}\bra{\psi_1}$, where $\ket{\psi_0}=\cos\frac{\pi}{8}\ket{0}+\sin\frac{\pi}{8}\ket{1},\ket{\psi_1}=\sin\frac{\pi}{8}\ket{0}+\cos\frac{\pi}{8}\ket{1}$. With this strategy, $\cS'^\cQ[\cJ'_1,\cJ'_2,\{\rho_0,\rho_1\},\{\sigma_x,\sigma_z\}]=3+\sqrt{2}>4=\cS'^\cC_{op}$.

\section{Condition for the figure of merit of the genuine instrument incompatibility certification tasks}
\label{App_genuine_cert}

\begin{proposition}
Consider a sequential communication task in which the figure of merit
depends only on the probabilities $\{p(b_y|x,y)\}_{b_y,x,y}$, namely
\begin{equation}
    \widehat{\cS}=\sum_{x,y,b_y}\alpha(x,y,b_y)p(b_y|x,y).
\end{equation}
Suppose that the measurements $M_y=\{\Phi^y_{b_y}(\Id_{\cK})\}$ induced by the instruments $\cI_y=\{\Phi^y_{b_y}\}^{d_y}_{b_y=1}\in\mathscr{I}(\cH,\cK)$ implemented by Bob are compatible where $y=\{0,1\}$. Then no quantum advantage over the classical sequential strategy with the same communication dimensions is possible.
\end{proposition}


\begin{proof}
    As the induced measurements $M_y=\{\Phi^y_{b_y}(\Id_{\cK})\}$ are compatible, there exists a joint measurement $M=\{M(b_0,b_1)\}$ such that $M_0(b_0)=\sum_{b_1}M(b_0,b_1)$, and $M_1(b_1)=\sum_{b_0}M(b_0,b_1)$. Then if Alice encodes her input $x$ in state $\rho_x$, we have
\begin{align}
    p(b_0|x,0)&=\tr[\Phi^0_{b_0}(\rho_x)]=\tr[\rho_xM_0(b_0)]\nonumber\\
    &=\sum_{b_1}\tr[\rho_xM(b_0,b_1)]=\sum_{b_1}p(b_0,b_1|x)
\end{align}
where $p(b_0,b_1|x):=\tr[\rho_xM(b_0,b_1)]$. Similarly,
\begin{align}
    p(b_1|x,1)&=\tr[\Phi^1_{b_1}(\rho_x)]=\tr[\rho_xM_1(b_1)]\nonumber\\
    &=\sum_{b_0}\tr[\rho_xM(b_0,b_1)]=\sum_{b_0}p(b_0,b_1|x).
\end{align}

However, note that the probability distribution $\{p(b_0,b_1|x)\}$ can be obtained using quantum strategy if Alice encodes her inputs in states $\{\rho_x\}$ Bob performs a single measurement $M$ on it. Hence, the probabilities of $\{p(b_y|x,y)\}\forall x,y$ can also be obtained using a classical strategy \cite{fw} in the scenario given in Fig. \ref{fig:communication_scenario} (a) where Charlie announces a fixed output. Furthermore, once Bob receives the input $y$, he can recover the $\{p(b_y|x,y)\}$ from $\{p(b_0,b_1|x)\}$ using classical post-processing of outcomes. Hence, the quantum advantage cannot be obtained for any chosen $\alpha(x,y,b_y)$. 
\end{proof}

\begin{proposition}
Consider a sequential communication task in which the figure of merit
depends only on the probabilities $\{p(c|x,y)\}_{c,x,y}$, namely
\begin{equation}
    \widetilde{S}=\sum_{x,y,c}\alpha(x,y,c)\,p(c|x,y).
\end{equation}
Suppose that the channels $\{\Phi^y=\sum^{d_y}_{b_y=1}\Phi^y_{b_y}\}_{y}$ induced by the instruments the instruments $\cI_y=\{\Phi^y_{b_y}\}^{d_y}_{b_y=1}\in\mathscr{I}(\cH,\cK)$ implemented by Bob are compatible where $y=\{0,1\}$, and that the dimension of the system communicated from
Bob to Charlie satisfies $d_B\geq d_c$. Then no quantum advantage over the classical sequential strategy with the same communication dimensions is possible.
\end{proposition}

\begin{proof}
For simplicity, consider two inputs $y\in\{0,1\}$ for Bob; the
generalization to an arbitrary number of inputs is immediate.
Let $M=\{M(c)\}_{c=0}^{d_c-1}$ be Charlie's measurement. Since
$\Phi^0$ and $\Phi^1$ are compatible, there exists a joint channel
$\Gamma\in\mathscr{C}(\cH,\cK_0\otimes\cK_1)$ with $\cK_0=\cK_1=\cK$ such that
\begin{equation}
    \Phi^0=\tr_{\cK_1}\circ\Gamma,
    \qquad
    \Phi^1=\tr_{\cK_0}\circ\Gamma .
\end{equation}
Consequently, the effective measurements
\begin{equation}
    E_{c|y}:=(\Phi^y)^\dagger[M(c)]
\end{equation}
are compatible. A parent measurement is given by
\begin{equation}
    G(c_0,c_1)
    :=
    \Gamma^\dagger\left(
        M(c_0)\otimes M(c_1)
    \right),
\end{equation}
so that
\begin{equation}
    E_{c_0|0}=\sum_{c_1}G(c_0,c_1),
    \qquad
    E_{c_1|1}=\sum_{c_0}G(c_0,c_1).
\end{equation}
Hence, for Alice's encoding states $\{\rho_x\}$, we have
\begin{equation}
    p(c|x,0)
    =
    \sum_{c_1}
    \tr\left[\rho_x G(c,c_1)\right], 
\end{equation}
and
\begin{equation}
    p(c|x,1)
    =
    \sum_{c_0}
    \tr\left[\rho_x G(c_0,c)\right]. 
\end{equation}
Therefore,
\begin{align}
    \widetilde{S}
    &=
    \sum_{x,c_0,c_1}
    \beta(x,c_0,c_1)
    \tr\left[\rho_x G(c_0,c_1)\right],
    \label{eq:single-povm-payoff}
\end{align}
where
\begin{equation}
    \beta(x,c_0,c_1)
    :=
    \alpha(x,0,c_0)\delta_{c,c_0}+\alpha(x,1,c_1)\delta_{c,c_1}.
\end{equation}
Equation~\eqref{eq:single-povm-payoff} is a standard
prepare-and-measure communication task in which Alice communicates
a $d_A$-dimensional system and the receiver performs a single fixed
measurement. By the result of Frenkel and Weiner~\cite{fw},
for every linear figure of merit the optimal value obtainable with
a $d_A$-dimensional quantum system can also be attained with a
$d_A$-valued classical message. Thus, there exists a classical
strategy described by $p(a|x)$, with $a\in[d_A]$, and
$q(c_0,c_1|a)$ which attains at least the value in
Eq.~\eqref{eq:single-povm-payoff}.

This classical prepare-and-measure strategy can be implemented in
the original sequential scenario. Alice sends $a$ to Bob. After
receiving $y$, Bob samples $(c_0,c_1)$ according to
$q(c_0,c_1|a)$ and sends the symbol
\begin{equation}
    m_B=c_y
\end{equation}
to Charlie. Since $d_B\geq d_c$, this encoding is allowed. Charlie
simply outputs $c=m_B$. Hence, the same value of the figure of merit
is attained by a classical sequential strategy with the prescribed
communication dimensions.
It follows that, whenever Bob's channels are compatible,
\begin{equation}
    \widetilde{S}^{Q_{\rm C}}\leq \widetilde{S}^{C}_{\rm op},
\end{equation}
and therefore, compatible channels cannot provide a quantum
advantage.
\end{proof}

\section{Certification of measurement incompatibility and channel incompatibility as a special case}

The certification of measurement incompatibility in measure-and-prepare scenario (i.e., the scenario depicted in Fig. \ref{fig:communication_scenario} (c) which is a special of Fig. \ref{fig:communication_scenario} (a)) has already been studied in Ref. \cite{Carmeli_2020,saha23}. Therefore, we will not discuss it here. Regarding channel incompatibility certification (i.e., the scenario depicted in Fig. \ref{fig:communication_scenario} (b) which is a special of Fig. \ref{fig:communication_scenario} (a)), consider the following figure of merit where $y,c,d_A,d_B\in\{0,1\}$ and $x_y\in\{0,1\}~\forall y$.

\bea
 \widetilde{\cS}=\frac{1}{8}\sum_{x,y,c}\delta_{c,x_y}p(c|x,y)=\frac{1}{8}\sum_{x,y}p(c=x_y|x,y) \label{Eq:chan_incomp_task}.
 \eea

As this figure of merit is similar to RAC, it can easily be shown that the classical upper bound is $\frac{3}{4}$. Furthermore, since the figue of merit depends on $\{p(c|x,y)\}$ with $d_B=d_c=2$ from the discussion of the above section, we conclude that if Bob uses two compatible channel, the quantum advantage cannot be achieved. Now, we show that there exists a pair of incompatible channels that provides quantum advantage. 

We know that the identity channel is incompatible with any other unitary channel \cite{Heinosaari_incomp_channel_2017}. Now, consider $x$ as a two bit string. Next, suppose, Bob is implementing the qubit idenitity channel $\Lambda_1=\mathbbm{I}_{\cH_Q}$ and a qubit unitary channel $\Lambda_2$ defined via its action on an arbitrary qubit state $\rho$, $\Lambda_2(\rho)=\textbf{H}\rho\textbf{H}$ for $y=0,1$ respectively where $\textbf{H}$ is Hadamard matrix, Charlie is performing the measurement $M=\{\ket{0}\bra{0},\ket{1}\bra{1}\}$, and Alice is encoding her input in the states $\{\rho_x=\ket{\psi_x}\bra{\psi_x}\}$ where four pure qubit states $\ket{\psi_{00}}=\cos\frac{\pi}{8}\ket{0}+\sin\frac{\pi}{8}\ket{1}$, $\ket{\psi_{01}}=\cos\frac{\pi}{8}\ket{0}-\sin\frac{\pi}{8}\ket{1}$, $\ket{\psi_{10}}=\sin\frac{\pi}{8}\ket{0}+\cos\frac{\pi}{8}\ket{1}$, and $\ket{\psi_{11}}=\sin\frac{\pi}{8}\ket{0}-\cos\frac{\pi}{8}\ket{1}$. Then it can be easily shown that the figure of merit given in Eq. \ref{Eq:chan_incomp_task} is $\frac{1}{2}(1+\frac{1}{\sqrt{2}})>\frac{3}{4}$. Hence, incompatibility of $\Lambda_1$ and $\Lambda_2$ is certified.

\end{document}